\documentclass[aps,pre,reprint,groupedaddress]{revtex4-1}

\pdfoutput=1
\usepackage{graphicx,amsmath}

\usepackage{xcolor}
\usepackage{braket}
\usepackage{xr}
\usepackage{amsthm}
\newtheorem{theorem}{Theorem}
\newtheorem{lemma}[theorem]{Lemma}
\newtheorem{property}[theorem]{Property}
\usepackage{mathtools}

\DeclarePairedDelimiter\floor{\lfloor}{\rfloor}

\begin{document}


\title{The mathematics of non-linear metrics for nested networks}



\author{Rui-Jie Wu, Gui-Yuan Shi, Yi-Cheng Zhang \& Manuel Sebastian Mariani}
\affiliation{Department of Physics, University of Fribourg, 1700 Fribourg, Switzerland}



\begin{abstract}
Numerical analysis of data from international trade and ecological networks has shown that the
non-linear fitness-complexity metric is the best candidate to rank nodes by importance
in bipartite networks that exhibit a nested structure.
Despite its relevance for real networks, the mathematical properties of the metric and its variants remain largely unexplored.
Here, we perform an analytic and numeric study of the fitness-complexity metric and a new variant, called minimal extremal metric.
We rigorously derive exact expressions for node scores for perfectly nested networks and show that these expressions explain
the non-trivial convergence
properties of the metrics.
A comparison between the fitness-complexity metric and the minimal extremal metric on real data
reveals that the latter can produce improved rankings if the input data
are reliable.
\end{abstract}


\maketitle

\onecolumngrid
\section{Introduction}

Network-based iterative algorithms are being applied to a broad range of problems,
such as ranking search results in the WWW \cite{brin1998anatomy}, predicting the traffic in urban roads \cite{jiang2008self},
recommending the items that an online user might appreciate \cite{lu2012recommender}, measuring the competitiveness of countries in world
trade \cite{hidalgo2009building, tacchella2012new},
ranking species according to their importance in plant-pollinator mutualistic networks \cite{allesina2009googling,dominguez2015ranking},
assessing scientific impact \cite{walker2007ranking,radicchi2009diffusion}, 
identifying influential spreaders \cite{ren2014iterative},
and many others.
While linear algorithms are applied to a broad range of
real systems \cite{gleich2015pagerank,ermann2015google},
it has been recently shown that the non-linear fitness-complexity metric introduced in ref. \cite{tacchella2012new}
markedly outperforms linear metrics in ranking the nodes by their importance
in bipartite networks that exhibit a nested architecture \cite{dominguez2015ranking, mariani2015measuring}.
The fitness-complexity metric has been originally introduced to rank countries and products in world trade
according to their level of competitiveness and quality, respectively \cite{tacchella2012new}.
The basic idea of the metric is that while the competitiveness of a country is mostly determined by the diversification of its exports,
the quality of a product is mostly determined by the score of the least competitive exporting countries.
The metric has been shown to be economically well-grounded \cite{tacchella2012new,cristelli2013measuring},
to be highly effective in ranking countries and products by their importance in the network \cite{mariani2015measuring},
to be informative about the future economic development \cite{cristelli2015heterogeneous}
and the future exports of a country \cite{vidmer2015prediction}.
The metric has been recently applied beyond its original scope and has been shown to be
the most efficient method among several network-based methods in ranking species according to their importance
in mutualistic ecological networks \cite{dominguez2015ranking}. In particular,
the metric reveals the nested structure of the system much
better than the methods used by standard nestedness calculators.
Several real systems exhibit a nested 
structure \cite{tacchella2012new, bascompte2003nested, jonhson2013factors, cimini2014scientific, borge2015nested, garas2015network},
which suggests that
the metric has a potentially broad range of application.


Despite the relevance of the fitness-complexity metric for nested networks,
its mathematical properties and its variants remain largely unexplored.
In contrast with linear algorithms such as Google's PageRank \cite{berkhin2005survey,gleich2015pagerank,ermann2015google}
and the method of reflections \cite{caldarelli2012network}, 
the convergence properties of the metric cannot be studied through linear algebra techniques.
This article provides new insights into the mathematics behind the metric.
We study here both the fitness-complexity metric (FCM) and a novel variant, called minimal extremal metric (MEM),
that is simpler to be treated analytically.
The only input of the metrics is the binary adjacency matrix $\mathcal{M}$ of the underlying bipartite network;
we perform here exact computations for perfectly nested matrices, i.e., binary matrices such that a unique border separates all
the elements equal to one from the elements equal to zero.
For both the MEM and the FCM, we find the exact analytic formulas that relate the ratios of node scores
to the shape of the underlying nested matrix.
While real nested matrices are not perfectly nested, the expressions derived here for perfectly nested matrices 
explain the non-trivial convergence properties the metrics found in real matrices \cite{pugliese2014convergence}.
In particular, we analytically determine the condition such that all node scores converge to a nonzero value, which
is crucial for the discriminative power of the metrics.
This condition has been also found in ref. \cite{pugliese2014convergence} (the only previous work that studied the convergence properties of the 
FCM); differently from the
analytic study of ref. \cite{pugliese2014convergence} where exact formulas were derived for matrices with two values of node score, 
in this work we derive by mathematical induction expressions valid for any perfectly nested matrix.

Finally, we contrast the two metrics in real data and show that the MEM can outperform the FCM in packing the adjacency matrix,
i.e., ordering its rows and columns in such a way that a sharp curve separates the occupied and empty regions of the matrix \cite{dominguez2015ranking}.
On the other hand, the MEM is more sensitive to noisy data, and, as 
a consequence, its rankings may be unreliable in the presence of a significant amount of mistakes in the 
original data \cite{battiston2014metrics}.

This article is organized as follows: In section \ref{algorithms}, we define the Fitness-Complexity metric (FCM) and the Minimal Extremal Metric (MEM);
In section \ref{compute}, we analytically compute the ratios between MEM and FCM node scores for perfectly nested matrices
and discuss the dependence of the metrics' convergence properties on the shape of the nested matrix;
In section \ref{realdata}, we compare the rankings by the FCM and the MEM in real data of world trade.

\section{Non-linear metrics for bipartite networks}
\label{algorithms}

In this section, we define the fitness-complexity metric (FCM) and the minimal extremal metric (MEM)
for bipartite networks.
While the results obtained in this paper hold for any nested matrix,
we use the terminology of economic complexity: rows and columns of the $N\times M$ adjacency matrix $\mathcal{M}$ are referred to as
countries and products, respectively; the matrix $\mathcal{M}$ is consequently referred to as 
the country-product matrix \cite{hidalgo2009building}.
We label countries by Latin letters ($i=1,\dots,N$),
products by Greek letters ($\alpha=1,\dots,M$); the number of countries and products are denoted by $N$ and $M$, respectively.
The number of products exported by country $i$ and the number of countries that export product
$\alpha$ are referred to as the diversification $D_{i}$ of country $i$ and the ubiquity $U_{\alpha}$
of product $\alpha$, respectively \cite{hidalgo2009building}.

In the fitness-complexity metric (FCM),
the fitness scores $\mathbf{F}=\{F_{i}\}$ of countries and complexity scores $\mathbf{Q}=\{Q_{\alpha}\}$ of products are
defined as the components of the fixed point of the following non-linear map \cite{tacchella2012new}
\begin{equation}
\begin{split}
  \tilde{F}_{i}^{(n)}&=\sum_{\alpha}\mathcal{M}_{i\alpha}Q_{\alpha}^{(n-1)},\\
  \tilde{Q}_{\alpha}^{(n)}&=\frac{1}{\sum_{i}\mathcal{M}_{i\alpha}\frac{1}{F^{(n-1)}_{i}}},
\end{split}
\label{fit_comp}
\end{equation}
where scores are normalized after each step according to
\begin{equation}
\begin{split}
  F_{i}^{(n)}=\tilde{F}_{i}^{(n)}/\,\overline{F^{(n)}},\\
  Q_{\alpha}^{(n)}=\tilde{Q}_{\alpha}^{(n)}/\,\overline{Q^{(n)}},
\end{split}
\end{equation}
with the initial condition $F_{i}^{(0)}=1$ and $Q_{\alpha}^{(0)}=1$.

Eq. \eqref{fit_comp} implies that the largest contribution to the complexity $Q$ of a product $\alpha$
is given by the fitness of the least-fit exporter
of product $\alpha$.
On the other hand, also the fitness scores of the other exporting countries contribute to $Q_{\alpha}$;
in this sense, the FCM is a quasi-extremal metric \cite{cristelli2013measuring}.
A natural question arises: how would the rankings change when modifying Eq. \eqref{fit_comp}, without
changing the main idea behind the metric?
A generalized version of the metric where the harmonic terms $1/F$ are replaced by $1/F^{\gamma}$, with $\gamma>0$, has 
been introduced in ref. \cite{pugliese2014convergence} and studied in refs. \cite{pugliese2014convergence,mariani2015measuring}.
Here, we introduce a simpler variant of the algorithm, called minimal extremal metric (MEM), where the complexity of a product is equal to the fitness
of the least-fit country that exports product $\alpha$.
This metric is extremal, which means that only $\min_{i:\mathcal{M}_{i\alpha}=1}{\{F_{i}^{(n)}\}}$ contributes to $Q_{\alpha}$.
In formulas
\begin{equation}
\begin{split}
  \tilde{F}_{i}^{(n)}&=\sum_{\alpha}\mathcal{M}_{i\alpha}Q_{\alpha}^{(n-1)},\\
  F_{i}^{(n)}&=\tilde{F}_{i}^{(n)}/\overline{F^{(n)}},\\
  Q_{\alpha}^{(n)}&=\min_{i:\mathcal{M}_{i\alpha}=1}{\{F_{i}^{(n)}\}}.
\end{split}
\label{mem}
\end{equation}
Note that the generalized FCM studied in Ref. \cite{mariani2015measuring} reduces to the MEM in the limit $\gamma\to\infty$.

\section{Analytic results}
\label{compute}

\begin{figure}[t]
 \centering
 \includegraphics[height=0.33\columnwidth,  angle=0]{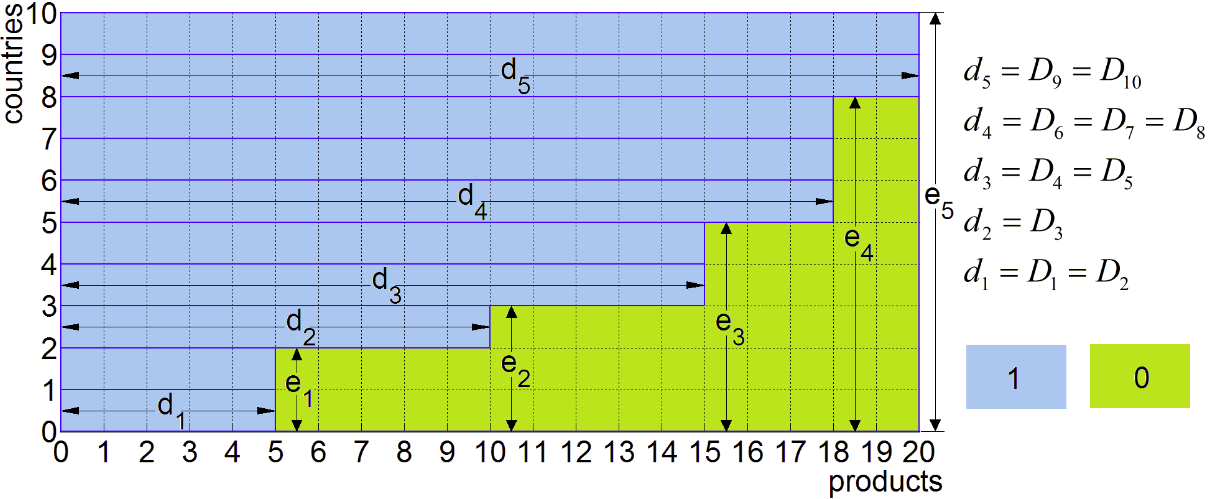}
 \caption{Illustration of the geometrical meaning of the variables $D,d,e$
 in a $10\times 20$ perfectly nested matrix. In this example, there are $m=5$ groups of countries, which correspond to the diversification values
 $d_{1},\dots,d_{5}$,
 and $m=5$ groups of products, which correspond to the ubiquity values $N,N-e_{1},\dots,N-e_4$.
 Due to the perfectly nested structure of the matrix, the groups of countries and products are in one-to-one correspondence: 
 the countries that belong to group $i$ are the least-fit exporters
 of the products belonging to group $i$, i.e., of the products whose ubiquity is $N-e_{i-1}$.
 \label{fig:matrix_explained}}
\end{figure}

\subsection{Perfectly nested matrix}
\label{pnm_def}

We focus here on perfectly nested matrices \cite{ulrich2007disentangling}, i.e., binary matrices 
where each country
exports all those products that are also exported by the less diversified countries plus a set of additional products.
Perfectly nested matrices are also known as stepwise matrices \cite{konig2011network}, and networks with a perfectly nested adjacency matrix
are also referred to as threshold networks \cite{hagberg2006designing}.
An example of perfectly nested matrix is shown in Fig. \ref{fig:matrix_explained}.
In the following, we label countries and products in order of increasing diversification ($D_{i+1}\geq D_{i}$)
and decreasing ubiquity, respectively ($U_{\alpha+1} \leq U_{\alpha}$).
We denote by $\Delta_{i}:=D_{i}-D_{i-1}$ the number of additional products that are exported by country $i$ but not by country $i-1$,
with $\Delta_1=D_1$.

According to Eqs. \eqref{fit_comp} and \eqref{mem}, countries with the same level of diversification have the same fitness score,
and it is thus convenient to group them together.
By doing this, we obtain $m$ groups of countries, with $m\leq N$; we denote by $d_{i}$ the diversification of countries that belong to group $i$,
where groups are labeled in order of increasing diversification and $i=1,\dots,m$. 
In addition, we denote by $e_i$ the number of countries whose diversification is smaller or equal than $d_{i}$.
This notation will turn out to be useful for the computations for the FCM.
We also define the number $\delta_i:=d_{i}-d_{i-1}$ of additional products that are exported by countries that belong to group $i$
but not by those belonging to group $i-1$, and the number $\epsilon_i:=e_{i}-e_{i-1}$ of countries that belong to group $i$ ($i=1,\dots,m$, and $e_0=d_0=0$).
Also products are divided into $m$ groups according to their level of ubiquity.
Since the number of country and product groups are the same and are equal to $m$, we use Latin letters ($i=1,\dots,m$) to label both groups.
Product groups are labeled in order of decreasing ubiquity; we denote by
$u_{i}=N-e_{i-1}$ the ubiquity of the products that belong to group $i$.
The geometrical interpretation of the variables $d,D,e$ is shown in Fig. \ref{fig:matrix_explained}.
Note that country and product groups are in one-to-one correspondence: countries that belong to group $i$ are the least-fit exporters of 
the products that belong to group $i$.

\subsection{Results for the MEM}
\label{mem_results}

For a perfectly nested matrix, the fitness of a country $i+1$ at iteration $n$ is given by the fitness
of country $i$ at iteration $n$, plus the complexity of the additional products that are exported by country $i+1$ but not by country $i$;
for the MEM, this property reads
\begin{equation}
  \tilde{F}_{i+1}^{(n+1)}=  F_{i}^{(n+1)} + F_{i+1}^{(n)}\times \Delta_{i+1},
  \label{fit_iter}
\end{equation}
where $F_{i+1}^{(n)}$ is the complexity of the additional products.
Our aim is to compute the ratios between the fitness scores.
We start by considering the two least-fit countries and compute the ratio $F_{1}^{(n+1)}/F_{2}^{(n+1)}$ between their scores.
Since we are only interested in the ratios between the fitness values, we do not normalize the variables $F$, $Q$ in the computation;
we have then $F^{(n)}_1=\Delta_{1}^{n}$ and, starting from Eq. \eqref{fit_iter}
\begin{equation}
\begin{split}
  \tilde{F}_{2}^{(n+1)} &= F_{1}^{(n+1)} + \Delta_{2}\times F_{2}^{(n)}\\
			&= F_{1}^{(n+1)} + \Delta_{2}\times (F_{1}^{(n)} + \Delta_{2} \times F_{2}^{(n-1)}) \\
			&= F_{1}^{(n+1)} + \Delta_{2}\times F_{1}^{(n)} + \Delta_{2}^{2} \times (F_{1}^{(n-1)} + \Delta_{2}\times F_{2}^{(n-2)})\\
			&= \Delta_{1}^{n+1} + \Delta_{2}\times \Delta_{1}^{n}+\Delta_{2}^{2}\times\Delta_{1}^{n-1}+\dots+\Delta_{2}^{n}\times\Delta_{1},
\end{split}
\label{rearranging}
\end{equation}
which can be rewritten as
\begin{equation}
\begin{split}
  \frac{F_{1}^{(n+1)}}{F_{2}^{(n+1)}}=\frac{1}{1+\frac{\Delta_{2}}{\Delta_{1}}+(\frac{\Delta_{2}}{\Delta_{1}})^{2}+\dots+(\frac{\Delta_{2}}{\Delta_{1}})^{n}}.
 \end{split}
\label{ratio}
\end{equation}
If $\Delta_{1}=\Delta_{2}$, 
\begin{equation}
\begin{split}
  \frac{F_{1}^{(n+1)}}{F_{2}^{(n+1)}}=\frac{1}{n+1} \underset{n\to\infty}{\longrightarrow} 0:
 \end{split}
\label{ratio3}
\end{equation}
the ratio converges to zero as $1/n$.
The ratio converges to zero also if $\Delta_{2}>\Delta_{1}$, but with an exponential rate:
\begin{equation}
\begin{split}
  \frac{F_{1}^{(n+1)}}{F_{2}^{(n+1)}}\simeq \Biggl(\frac{\Delta_{1}}{\Delta_{2}} \Biggr)^{n}  \underset{n\to\infty}{\longrightarrow} 0.
 \end{split}
\label{ratio2}
\end{equation}
By contrast, using the geometric series we can show that the ratio is finite if $\Delta_{2}<\Delta_{1}$:
\begin{equation}
\begin{split}
  \frac{F_{1}^{(n+1)}}{F_{2}^{(n+1)}} \underset{n\to\infty}{\longrightarrow} 1-\frac{\Delta_{2}}{\Delta_{1}}.
 \end{split}
\label{ratio1}
\end{equation}
Interestingly, the three different asymptotic behaviors \eqref{ratio3}, \eqref{ratio2}, and \eqref{ratio1} 
correspond to the asymptotic behaviors found in ref. \cite{pugliese2014convergence} for the FCM fitness scores 
with a model matrix where there are only two values $F_{1}$ and $F_{2}$ of fitness score.
We will now use this result as the starting point of
a rigorous derivation of the analytic expression for the fitness ratios in an arbitrary perfectly nested matrix.

First, we note that Eqs. \eqref{ratio1} and \eqref{ratio2} can be summarized as
\begin{equation}
\lim_{n\to\infty}\frac{F_{1}^{(n)}}{F_{2}^{(n)}}=  1-\frac{\Delta_{2}}{\max{\{\Delta_{1},\Delta_{2}\}}}.
\label{solution_minimal_prov}
\end{equation}
Starting from Eq. \eqref{fit_iter} and using mathematical induction, we can show that (see Appendix A)
\begin{equation}
\lim_{n\to\infty}\frac{F^{(n)}_{i}}{F^{(n)}_{i+1}}=1-\frac{\Delta_{i+1}}{\max_{j\in[1,i+1]}{\{\Delta_{j}\}}}.
\label{solution_minimal}
\end{equation}
Note that Eq. \eqref{solution_minimal} relates the score ratios $F^{(n)}_{i}/F^{(n)}_{i+1}$ 
to the shape of the perfectly nested matrix, which is encoded in the set of the $\Delta$ values.
Eq. \eqref{solution_minimal} implies that for any
perfectly nested matrix $\mathcal{M}$, all MEM fitness scores converge to a
nonzero value if and only if $\Delta_{i}<\Delta_{1} \, \forall \, i>1$.
If the gap $\Delta_{i+1}$ between the diversifications of countries $i$ and $i+1$ is the largest gap among the gaps $\Delta_j$ of the countries $j\leq i+1$,
then the ratio between the score of country $i$ and the score of country $i+1$ converges to zero.
The derivation of Eq. \eqref{solution_minimal} is a first example of how the behavior of non-linear metrics can be completely understood for
perfectly nested matrices;
in the next section, we will derive an analogous expression for the FCM.

Eq. \eqref{solution_minimal} suggests that for a matrix that is not too different from a perfectly nested matrix,
the score ratios could be used to assess the convergence of the metric. 
In particular, one can decide to halt iterations when the following criterion is met:
\begin{equation}
 \sum_{i=1}^{N} \Bigg| \frac{F_{i}^{(n)}}{F_{i+1}^{(n)}}-\frac{F_{i}^{(n+1)}}{F_{i+1}^{(n+1)}}\Bigg|<\epsilon,
 \label{convergence_criterion_main}
\end{equation}
where $\epsilon$ is a predefined accuracy threshold.
We refer to \ref{convergence} for the results of the application of this criterion to real data,
and to \ref{complexity}
for a numerical study of the dependence of the convergence iteration on the size of the system.
As first suggested by ref. \cite{pugliese2014convergence},
if some score ratios converge to zero, countries can be naturally separated in different groups for which all fitness ratios
converge to a nonzero value within a set. We refer to \ref{submatrices} for a real example from world trade
of country separation implied by the existence of zero fitness ratios.

\subsection{Results for the FCM}

The FCM score of a certain product is determined by the scores of all the exporting countries,
which makes the analytical computations for the FCM more difficult than those for the MEM.
For the computations with the FCM, it is convenient to group together countries with the same level of diversification.
In agreement with the definitions provided in paragraph \ref{pnm_def}, we denote by $f_{i}$ the fitness of countries that belong to group $i$, 
i.e., of those countries whose diversification is equal to $d_{i}$.
Analogously, we denote by $q_{i}$ the complexity of the products whose least-fit exporting countries belong to group $i$.
We have then $m$ fitness scores $\{f^{(n)}_{1},\dots, f^{(n)}_{m}\}$ and $m$ complexity scores $\{q^{(n)}_{1},\dots, q^{(n)}_{m}\}$.
With this notation, we rewrite Eq. \eqref{fit_comp} as
\begin{equation}
 \begin{split}
  \tilde{f}_{i}^{(n)}&=\sum_{j=1}^{i}\delta_{j}\,q_{j}^{(n-1)},\\
  \tilde{q}_{i}^{(n)}&=\frac{1}{\sum_{j=i}^{m}\epsilon_j/f^{(n)}_{j}},
\end{split}
\label{fit_comp_ptm}
\end{equation}
where in the r.h.s. of the second line we replaced $1/f^{(n-1)}_{j}$ with $1/f^{(n)}_{j}$, which does not affect the results
in the limit $n\to\infty$.
Note that in the r.h.s. of the second line, the factor $\epsilon_j$ of the terms $1/f^{(n)}_{j}$ represents the
number of countries that belong to group $j$.
Now we transform Eqs. \eqref{fit_comp_ptm} into a set of equivalent equations
for the fitness ratio $x_{i}^{(n)}:=f_{i}^{(n)}/f_{i+1}^{(n)}$ and the complexity ratio $y_{i}^{(n)}:=q_{i}^{(n)}/q_{i+1}^{(n)}$.
The equation that relates the scores of two consecutive countries $i$ and $i+1$ is
\begin{equation}
 f_{i+1}^{(n+1)}=f_{i}^{(n+1)}+q_{i+1}^{(n)}\times \delta_{i+1}.
 \label{fit_iter_2}
\end{equation}
In terms of the $x$ variables, Eq. \eqref{fit_iter_2} is equivalent to
\begin{equation}
 \frac{1}{x_{i}^{(n)}}=1+\frac{\delta_{i+1}\,q_{i+1}^{(n-1)}}{\tilde{f}_{i+1}^{(n)}},
 \label{xeq}
\end{equation}
which implies
\begin{equation}
 \frac{1/x_{i}^{(n)}-1}{1/x_{i-1}^{(n)}-1}=\frac{\delta_{i+1}}{\delta_{i}}\,\frac{x_{i-1}^{(n)}}{y_{i}^{(n-1)}};
\end{equation}
reshuffling the terms of this equation, we get
\begin{equation}
 x_{i}^{(n)}=\frac{\delta_{i}\,y_{i}^{(n-1)}}{\delta_{i}\,y_{i}^{(n-1)}+\delta_{i+1}\,(1-x_{i-1}^{(n)})}.
 \label{equivalent1}
\end{equation}
For the least-fit country ($i=1$), from Eq. \eqref{xeq} we directly obtain
\begin{equation}
 x_{1}^{(n)}=\frac{\delta_{1}\,y_{1}^{(n-1)}}{\delta_{1}\,y_{1}^{(n-1)}+\delta_{2}}.
 \label{equivalent2}
\end{equation}
Starting from the second line of Eq. \eqref{fit_comp_ptm} and proceeding in a similar way, we obtain the analogous equations for the $y$ variable:
\begin{equation}
 y_{i}^{(n)}=\frac{\epsilon_{i+1}x_{i}^{(n)}}{\epsilon_{i+1}x_{i}^{(n)}+\epsilon_{i}(1-y_{i+1}^{(n)})}
 \label{equivalent3}
\end{equation}
and
\begin{equation}
 y_{m-1}^{(n)}=\frac{\epsilon_{m}x_{m-1}^{(n)}}{\epsilon_{m}x_{m-1}^{(n)}+\epsilon_{m-1}}.
 \label{equivalent4}
\end{equation}
The set formed by Eqs. \eqref{equivalent1},  \eqref{equivalent2},  \eqref{equivalent3}, \eqref{equivalent4} is exactly
equivalent to the original fitness-complexity equations (Eq. \eqref{fit_comp_ptm}).
The uniform initial condition $f_{i}^{(0)}=1\,\,\, \forall i$ implies the initial conditions
\begin{equation}
 x_{i}^{(0)}=1,
\end{equation}
\begin{equation}
 y_{i}^{(0)}=\frac{e_m-e_i}{e_m-e_{i-1}}
\end{equation}
for the variables $x$ and $y$.
Using Eqs. \eqref{equivalent1}-\eqref{equivalent4}, we prove the following lemma:

\begin{lemma}[Convergence]
 The sequences $\{x_{i}^{(n)}\}$ and $\{y_{i}^{(n)}\}$ are convergent in the limit $n\to\infty$.
\end{lemma}

\begin{proof}
 To prove the convergence, we first prove that the sequences $\{x_{i}^{(n)}\}$ and $\{y_{i}^{(n)}\}$ are decreasing in $n$.
From Eq. \eqref{equivalent2}, we have
\begin{equation}
 x_{1}^{(1)}=\frac{\delta_{1}\,y_{0}^{(1)}}{\delta_{1}\,y_{0}^{(1)}+\delta_{2}}<1=x_{1}^{(0)};
\label{first_inequality}
 \end{equation}
by combining inequality \eqref{first_inequality} with Eq. \eqref{equivalent1}, we get $x_{2}^{(1)}<x_{2}^{(0)}$;
we can repeat the same for all $i$ and get
\begin{equation}
x_{i}^{(1)}<x_{i}^{(0)} \,\,\forall i.
\label{inequality_x}
\end{equation}
Analogously, by combining inequality \eqref{inequality_x} with Eq. \eqref{equivalent4} we get $y_{m-1}^{(1)}<y_{m-1}^{(0)}$,
from which we can iteratively show that
\begin{equation}
y_{i}^{(1)}<y_{i}^{(0)} \,\,\forall i.
\label{inequality_y}
\end{equation}
Now, we use mathematical induction on $n$ to prove that $x_i^{(n+1)}<x_i^{(n)}$ and $y_i^{(n+1)}<y_i^{(n)}$, for all $i=1,\dots,m-1$. 
Suppose that $x_i^{(n)}<x_i^{(n-1)}$ and $y_i^{(n)}<y_i^{(n-1)}$.
From Eq. \eqref{equivalent2}, the former inequality directly implies
\begin{equation}
 x_{1}^{(n+1)}=\frac{\delta_{1}\,y_{1}^{(n)}}{\delta_{1}\,y_{1}^{(n)}+\delta_{2}}
 < \frac{\delta_{1}\,y_{1}^{(n-1)}}{\delta_{1}\,y_{1}^{(n-1)}+\delta_{2}}=x_{1}^{(n)}.
\end{equation}
To prove the inequality $x_{i}^{(n+1)}<x_{i}^{(n)}$ for all $i$, we use mathematical induction on $i$. To do this,
we show that $x_{i-1}^{(n+1)}<x_{i-1}^{(n)}$ implies $x_{i}^{(n+1)}<x_{i}^{(n)}$.
We obtain
\begin{equation}
 x_{i}^{(n+1)}=\frac{\delta_{i}\,y_{i}^{(n)}}{\delta_{i}\,y_{i}^{(n)}+\delta_{i+1}\,(1-x_{i-1}^{(n+1)})}
 < \frac{\delta_{i}\,y_{i}^{(n-1)}}{\delta_{i}\,y_{i}^{(n-1)}+\delta_{i+1}\,(1-x_{i-1}^{(n+1)})}
 < \frac{\delta_{i}\,y_{i}^{(n-1)}}{\delta_{i}\,y_{i}^{(n-1)}+\delta_{i+1}\,(1-x_{i-1}^{(n)})}=x_{i},
\end{equation}
where we used the induction hypothesis on $n$ in the first inequality, and the induction hypothesis on $i$ in the last inequality.
A similar proof can be carried out to get $y_i^{(n+1)}<y_i^{(n)}$.
Since $x_{i}^{(n)}$ and $y_{i}^{(n)}$ are decreasing sequences in $n$ and $x_{i}^{(n)}>0, y_{i}^{(n)}>0 \,\, \forall n$, 
then $x_{i}^{(n)}$ and $y_{i}^{(n)}$ converge when $n\to\infty$. 
\end{proof}

\subsubsection{Score ratios when the diagonal does \emph{not} cross the empty region of $\mathcal{M}$}

The lemma ensures the convergence of the non-linear map defined by Eq. \eqref{fit_comp_ptm}.
We use now the lemma to prove the theorem that guarantees the convergence of the score ratios to a unique fixed point.
The theorem holds when the diagonal of the matrix does not cross the empty region of the matrix $\mathcal{M}$, i.e.,
the region whose elements are zero.
In formulas, this condition reads
\begin{equation}
 d_i>e_i\,\frac{d_m}{e_m} \,\,\, \forall i=1,\dots,m-1.
 \label{nocrossing1}
\end{equation}
We will also discuss then the procedure to compute the fitness ratios when condition $\eqref{nocrossing1}$ is false.
We emphasize that Ref. \cite{pugliese2014convergence} found this property through an analytic computation on theoretical matrices where
two values $F_1$ and $F_2$ of fitness score are present, conjectured its validity for any nested matrix and verified this 
hypothesis through numerical simulations on several datasets.
Here, we demonstrate its validity for any perfectly nested matrix.

 \begin{figure}[t]
 \centering
 \includegraphics[height=0.33\columnwidth,  angle=0]{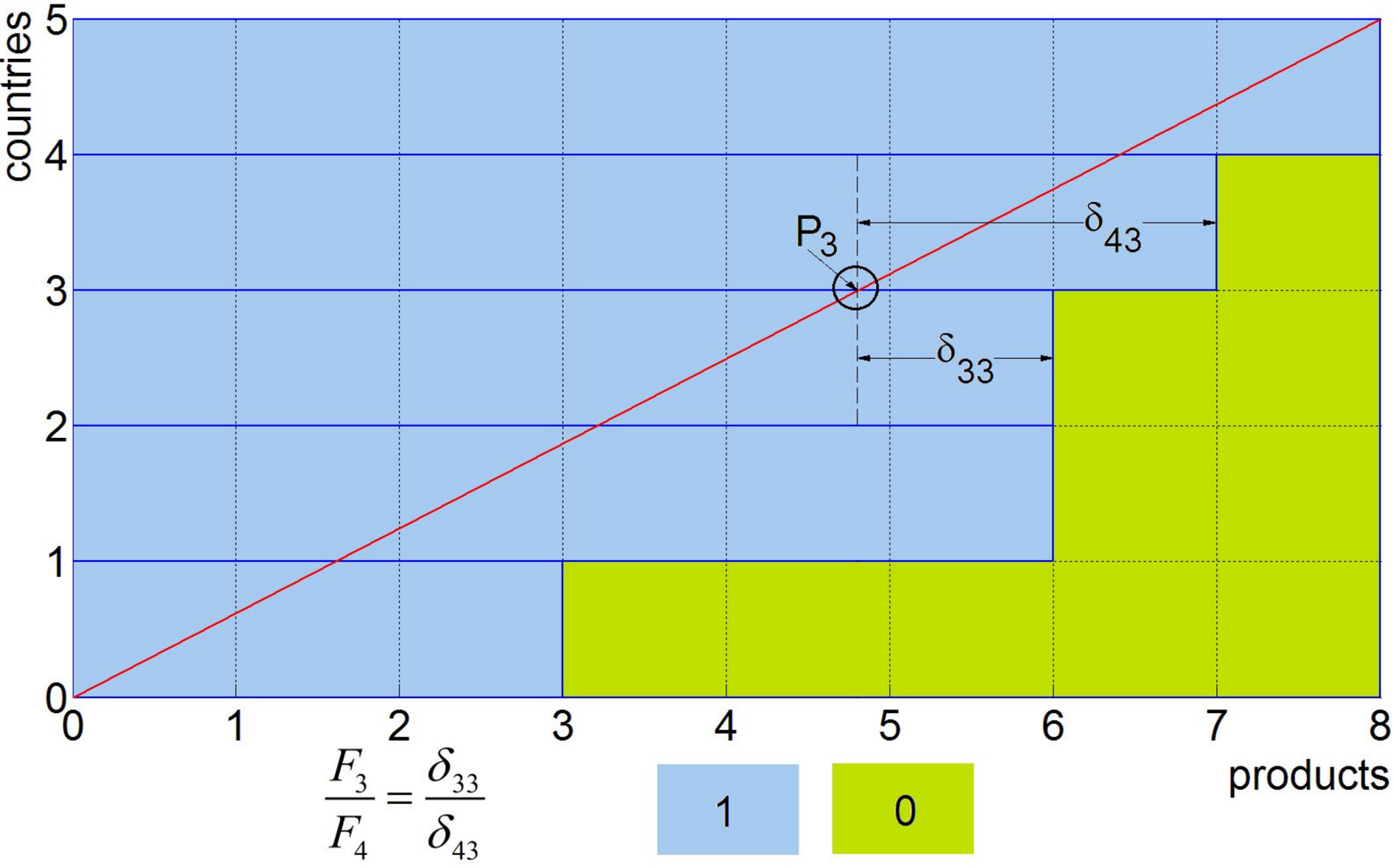}\includegraphics[height=0.33\columnwidth,  angle=0]{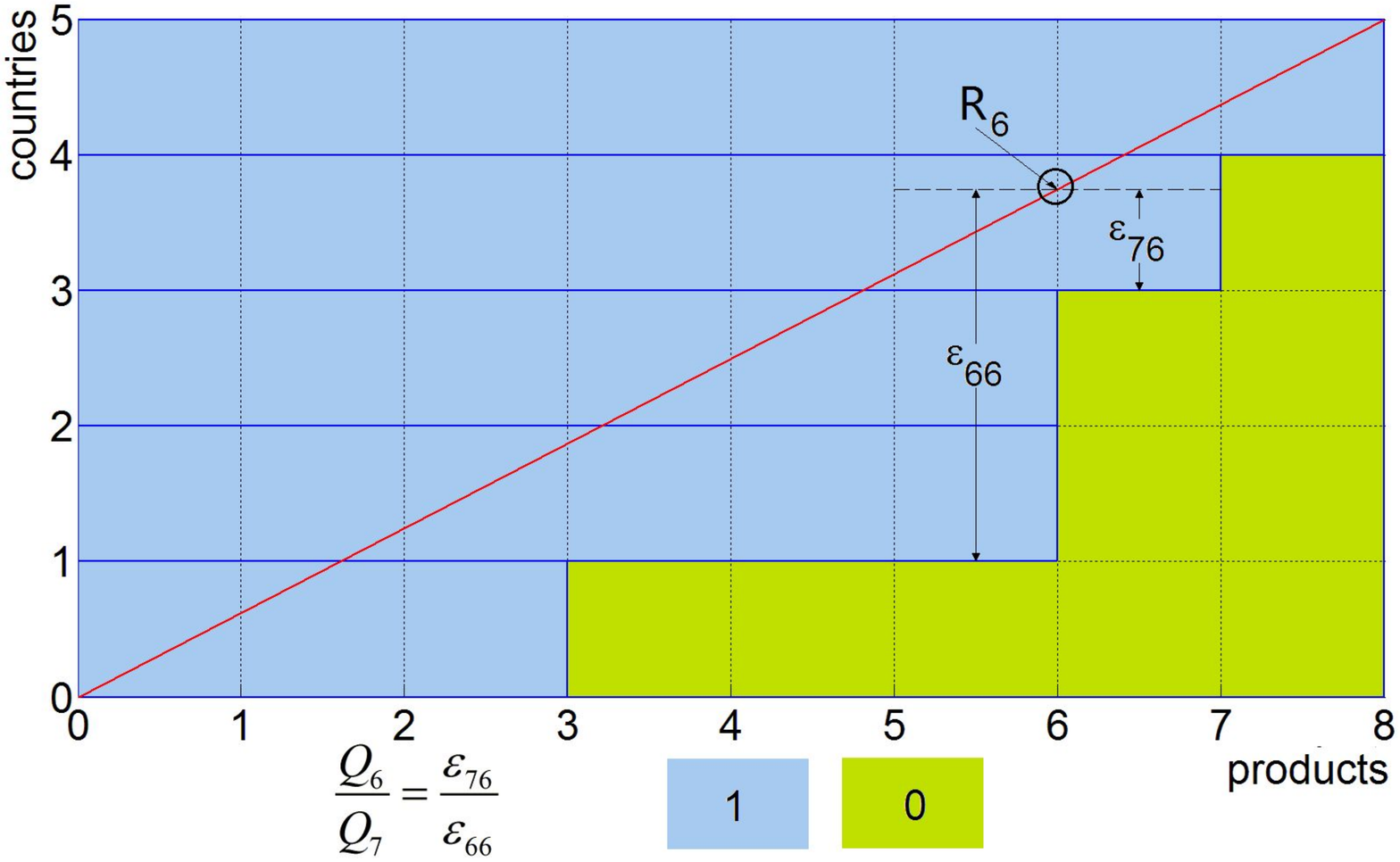}
 \caption{Illustration of the formulas \eqref{sol1}-\eqref{sol2} for computing the score ratios in a $5\times 8$ matrix where
 the diagonal never crosses the empty region of the matrix.
 We denote by $\delta_{ij}:=D_{j}-P_{ix}$ the distance between the point $P_i$ where the diagonal of the matrix intersect the line $y=i$
 and the line $x=D_j$.
 We have then $F_{3}/F_{4}=\delta_{33}/\delta_{43}$.
 Analogously, we denote by $\epsilon_{ij}:=R_{iy}-E_{j}$ the distance between the point $R_i$ where the diagonal of the matrix intersect the line $x=i$ and
 the line $y=E_{j}=N-U_j$ 
 We have then $Q_{6}/Q_{7}=\epsilon_{76}/\epsilon_{66}$. 
 \label{fig:solution1}}
\end{figure}

 \begin{figure}[t]
 \centering
 \includegraphics[height=0.33\columnwidth,  angle=0]{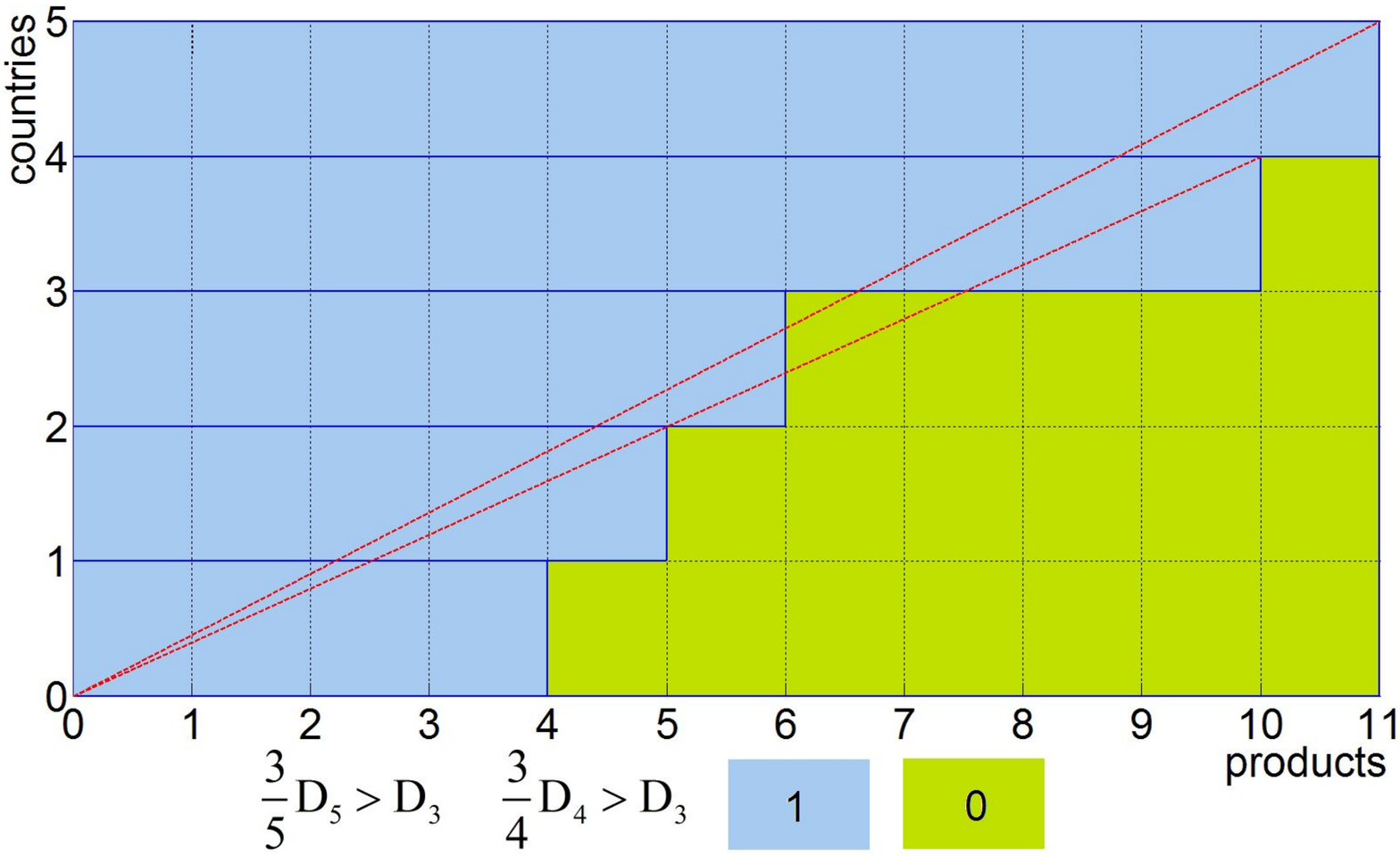}\includegraphics[height=0.33\columnwidth,  angle=0]{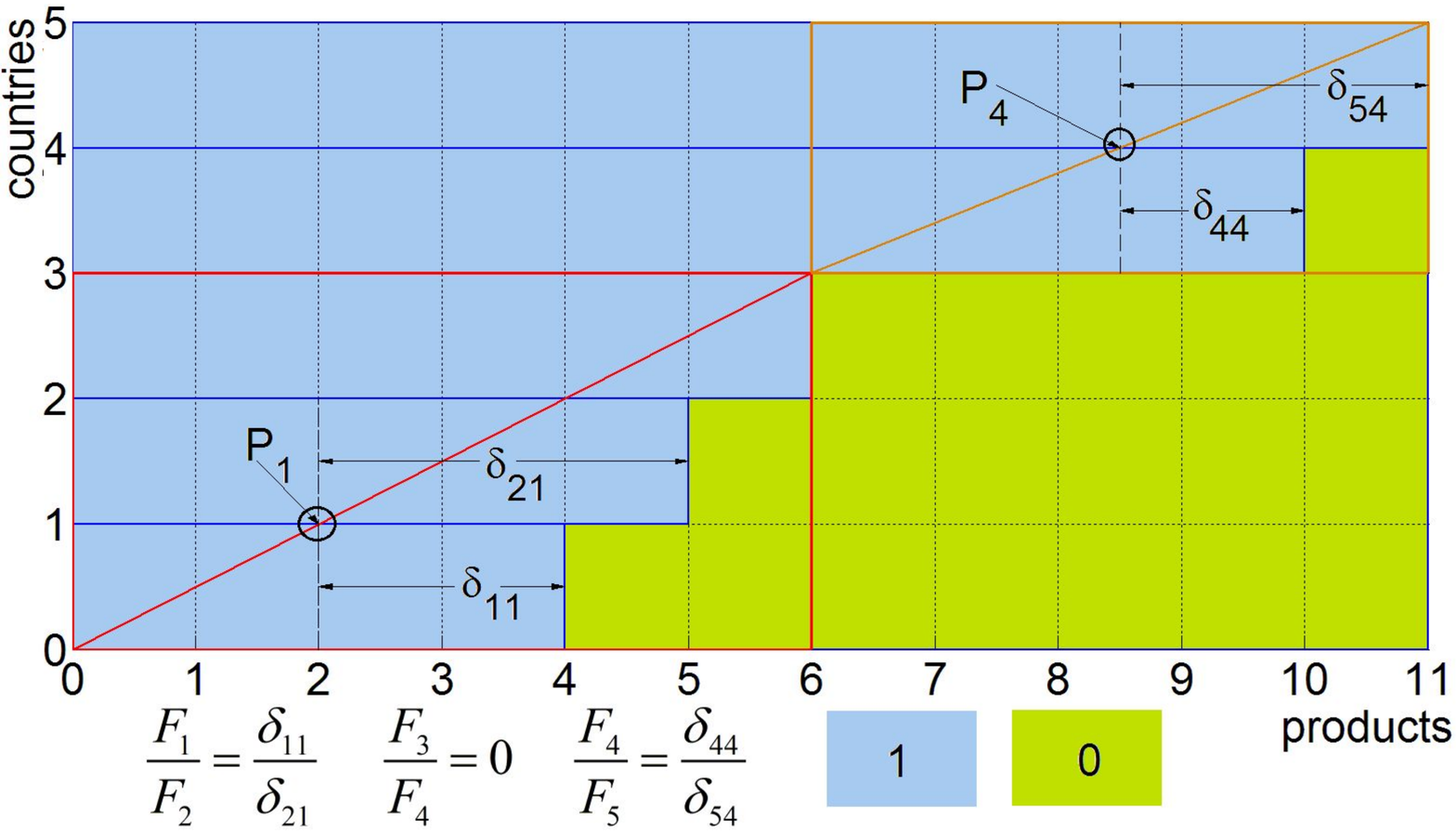}
 \caption{Illustration of the procedure for computing the score ratios in a $5\times 8$ matrix where the diagonal crosses
 the empty region of the matrix.
 In this case, we have to find the most-fit country $j_{max}$ such that condition \eqref{condition} holds. In this example,
 $j_{max} \neq 5$ and $j_{max}\neq 4$, because the diagonal from $(0,0)$ to $(d_j,j)$ crosses the empty region of the matrix for $j=4,5$
 (left panel). We find $j_{max}=3$. We can then compute the score ratios for all the countries $i\leq 3$ and all the products $\alpha\leq d_{3}=6$.
 To do this, we use the same geometrical construction of Fig. \ref{fig:solution1}, but restricted to the submatrix that contains only the three least-fit 
 countries and the $d_{3}$ most-ubiquitous products, which corresponds to the block delimited by red border in the right panel.
 We can then remove from the matrix the countries $i\leq 3$ and the products $\alpha<d_{3}=6$.
 and compute the score ratios for the countries and products in the residual matrix, 
 which corresponds to the block delimited by orange border in the right panel.
 \label{fig:solution2}}
\end{figure}

\begin{theorem} If condition \eqref{nocrossing1} holds,
 then
 \begin{equation}
 \lim_{n\to\infty}  \frac{f_{i}^{(n)}}{f_{i+1}^{(n)}} = a_{i},
  \label{solution_fitness_1}
 \end{equation}
  \begin{equation}
  \lim_{n\to\infty}  \frac{q_{i}^{(n)}}{q_{i+1}^{(n)}} = b_{i},
  \label{solution_complexity_1}
 \end{equation}
 where 
 \begin{equation}
 a_{i}= \frac{d_{i}-\frac{d_{m}}{e_m}e_i}{d_{i+1}-\frac{d_{m}}{e_m}e_i},
 \label{a_def}
 \end{equation}
 \begin{equation}
   b_{i} = \frac{\frac{e_{m}}{d_m}d_i-e_i}{\frac{e_{m}}{d_m}d_i-e_{i-1}}.
\end{equation}
\end{theorem}

We refer to Appendix B for the details of the proof.
The components of the limit vectors $(\mathbf{a},\mathbf{b})$ have a simple geometrical interpretation.
To see this, we rewrite Eqs.\eqref{solution_fitness_1}-\eqref{solution_complexity_1} in terms of the original variables
$F,Q,D,U$:
\begin{equation}
\lim_{n\to+\infty}\frac{F^{(n)}_i}{F^{(n)}_{i+1}}=\frac{D_i-\frac{M}{N}i}{D_{i+1}-\frac{M}{N}i},
\label{sol1_ori}
\end{equation}
\begin{equation}
\lim_{n\to+\infty}\frac{Q^{(n)}_i}{Q^{(n)}_{i+1}}=\frac{\frac{N}{M}i-E_{i+1}}{\frac{N}{M}i-E_{i}},
\label{sol2_ori}
\end{equation}
where we defined $E_{i}=N-U_{i}$.
In term of the original variables, condition \eqref{nocrossing1} reads 
\begin{equation}
 D_i>\frac{M}{N}i
 \label{condition_nocrossing}
\end{equation} 
The solution \eqref{sol1_ori}-\eqref{sol2_ori} has
a simple geometric interpretation when considering the representation of the matrix $\mathcal{M}$ in the euclidean plane.

If we denote by $P_{ix}$ the $x$-coordinate of the point $P_{i}$ where the diagonal
of the matrix -- i.e., the diagonal from $(0,0)$ to $(M,N)$ --
 intersects the horizontal line $y=i$, we have exactly $P_{ix}=i\,M/N$ (see Fig. \ref{fig:solution1}).
 As a consequence, assuming $D_i>i\,M/N$ is equivalent to assuming that
 the diagonal of the matrix never crosses the empty region of the matrix.
 Eq. \eqref{sol1} can thus be rewritten as
 \begin{equation}
  \frac{F_{i}}{F_{i+1}}=\frac{D_{i}-P_{ix}}{D_{i+1}-P_{ix}}.
  \label{sol1}
 \end{equation}
 As shown in Fig. \ref{fig:solution1}, the numerator and the denominator can be interpreted as
 the distances of the point $P_i$ from the vertical lines $x=D_i$ and $x=D_{i+1}$, respectively.
One can also show that condition 
\eqref{condition_nocrossing} implies $M\,E_{i+1}<i\,N$ ($i=1,\dots,M-1$), If we denote by $R_{iy}$ the $y$-coordinate of the point $R_{i}$ 
where the diagonal from $(0,0)$ to $(M,N)$ intersects the line $x=i$, we have $R_{iy}=i\,N/M$. Eq. \eqref{sol2}
can be rewritten as:
 \begin{equation}
  \frac{Q_{i}}{Q_{i+1}}=\frac{R_{iy}-E_{i+1}}{R_{iy}-E_{i}},
  \label{sol2}
 \end{equation}
 which has a simple geometrical interpretation as well (see Fig. \ref{fig:solution1}).

 \subsubsection{Score ratios when the diagonal does cross the empty region of $\mathcal{M}$}
 \label{analytics_crossing}
 
If the diagonal of the matrix crosses the empty region of the matrix -- i.e., if there exists some $i$ such that $d_i\leq e_i\,d_m/e_m$ --
we cannot directly use Eqs. \eqref{solution_fitness_1}, \eqref{solution_complexity_1}.
In this case, the procedure to compute the fitness and complexity ratios is the following:
\begin{enumerate}
 \item We find the most-fit country $j_{max}$ such that
\begin{equation}
D_{i}-\frac{i}{j_{max}}\,D_{j_{max}}>0 \, \forall \, i\leq j_{max}.
 \label{condition}
\end{equation}
 When considering the matrix $\mathcal{M}$ in the euclidean plane, the country $j_{max}$ corresponds to
 the most-fit country such that the diagonal from $(0,0)$ to $(d_{j_{max}},j_{max})$ never crosses the empty region of the reduced matrix
 that contains only the countries $j<j_{max}$,
 as shown in Fig. \ref{fig:solution2}.
 \item Once the value of $j_{max}$ has been determined, we can compute all the fitness ratios for the countries $i<j_{max}$ as
 \begin{equation}
  \frac{F_{i}}{F_{i+1}}=\frac{D_{i}-iD_{j_{max}}/j_{max}}{D_{i+1}-iD_{j_{max}}/j_{max}},
  \label{solution_fitness}
 \end{equation}
  \begin{equation}
  \frac{Q_{i}}{Q_{i+1}}= \frac{ij_{max}/D_{j_{max}}-E_{i+1}}{ij_{max}/D_{j_{max}}-E_{i}}.
  \label{solution_complexity}
 \end{equation}
 for all $i<j_{max}$, and $F_{j_{max}}/F_{j_{max}+1}=0$, $Q_{j_{max}}/Q_{j_{max}+1}=0$.
 Note that this formula has the same geometrical interpretation of Eqs. \eqref{solution_fitness_1}-\eqref{solution_complexity_1},
 but the geometrical construction is carried out in a submatrix of the matrix $\mathcal{M}$ (see Fig. \ref{fig:solution2}).
 \item We remove from the matrix all the countries $i\leq j_{max}$ and all the products $\alpha\leq d_{j_{max}}$, 
 and restart from point 1, until
 all the ratios are computed.
\end{enumerate}

The interpretation of this procedure is simple: if the diagonal line crosses the empty region of the matrix,
there is at least one pair of countries for which the score ratio converges to zero.
In this case, the matrix should be split in blocks such that the score ratios are all nonzero within each block;
the score ratios can then be computed inside each block according to Eqs. \eqref{solution_fitness}-\eqref{solution_complexity}.
A graphical illustration of this procedure is shown in Fig. \ref{fig:solution2}.

\section{Results in real networks}
\label{realdata}

 \begin{figure}[t]
 \centering
 \includegraphics[height=0.45\columnwidth,  angle=270]{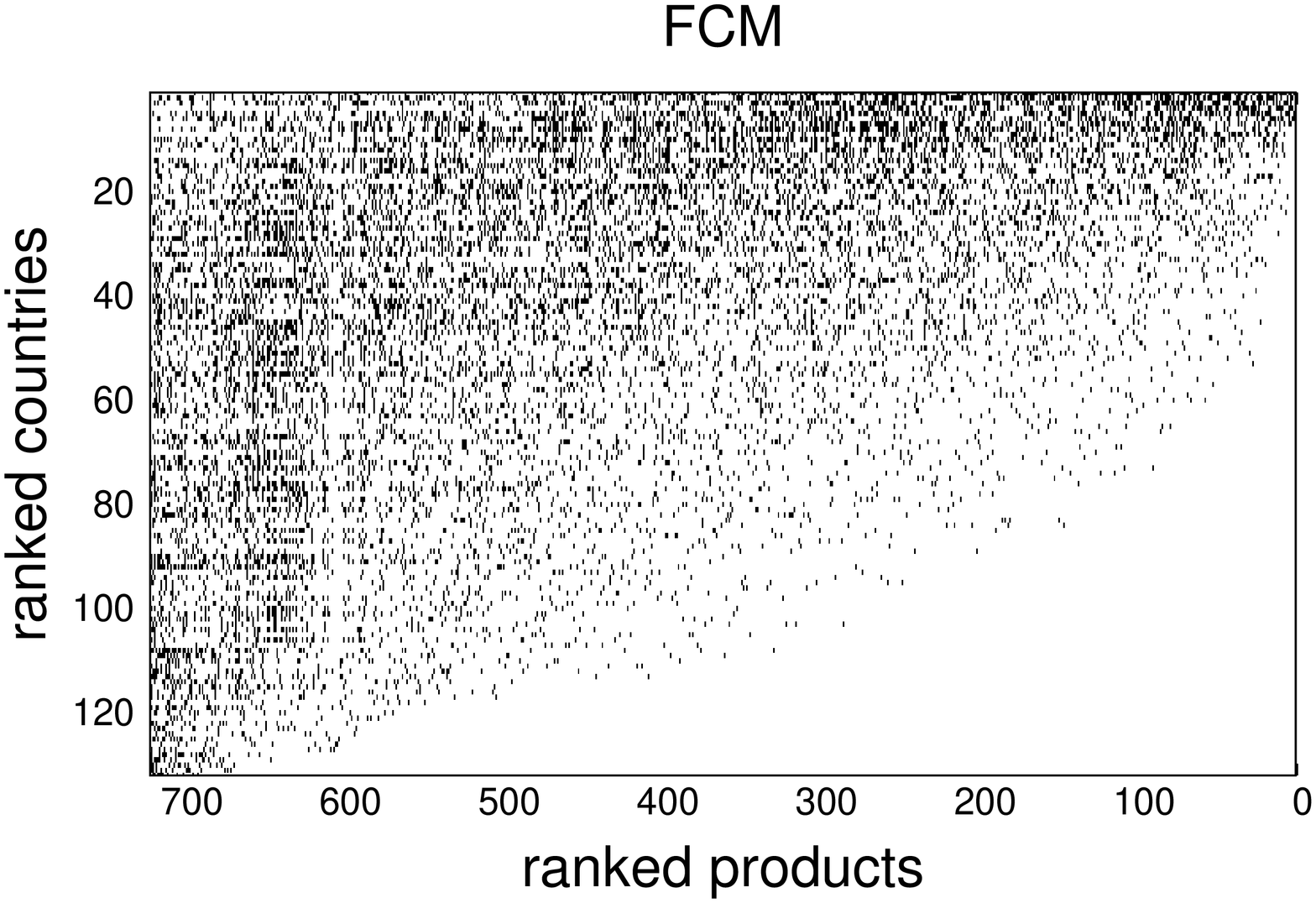}\includegraphics[height=0.45\columnwidth,  angle=270]{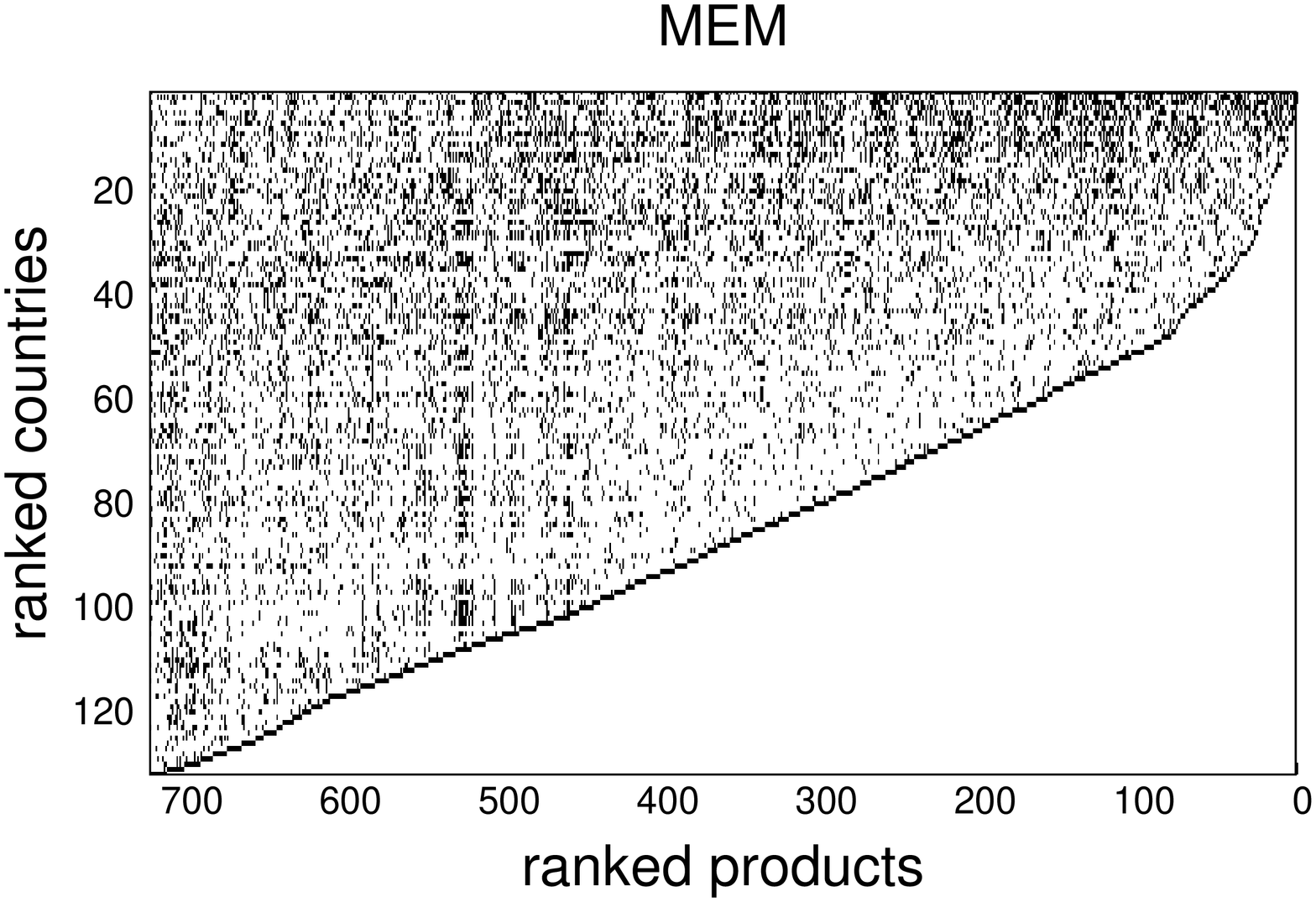}
 \caption{Country-product matrices resulting from the FCM and the MEM (1996).
Both matrices are nested, but the border
between the filled and the empty region of the matrix is sharper for the MEM than for the FCM. 
 \label{fig:matrix}}
\end{figure}

 \begin{figure}[t]
 \centering
 \includegraphics[height=0.45\columnwidth,  angle=270]{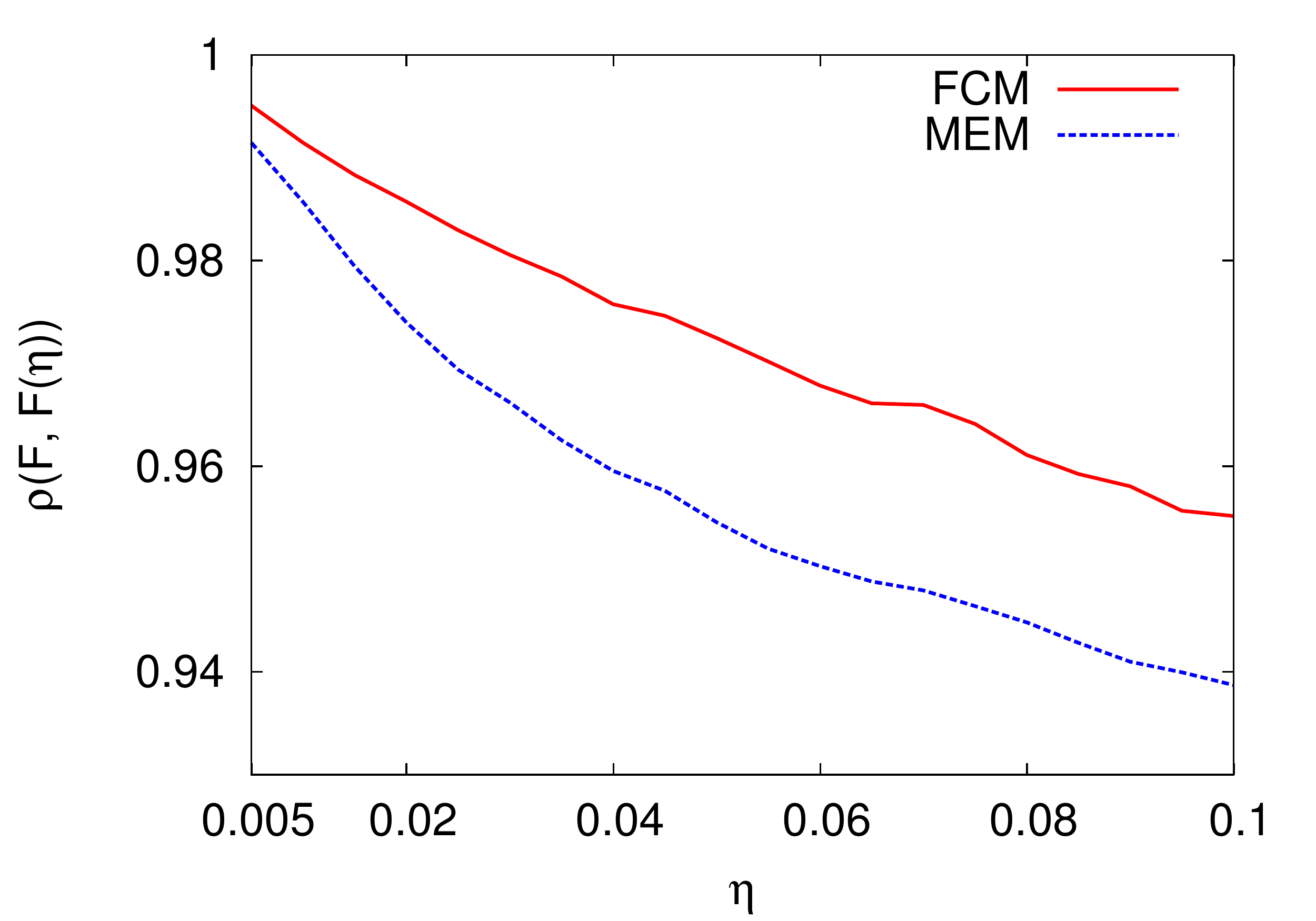}\includegraphics[height=0.45\columnwidth,  angle=270]{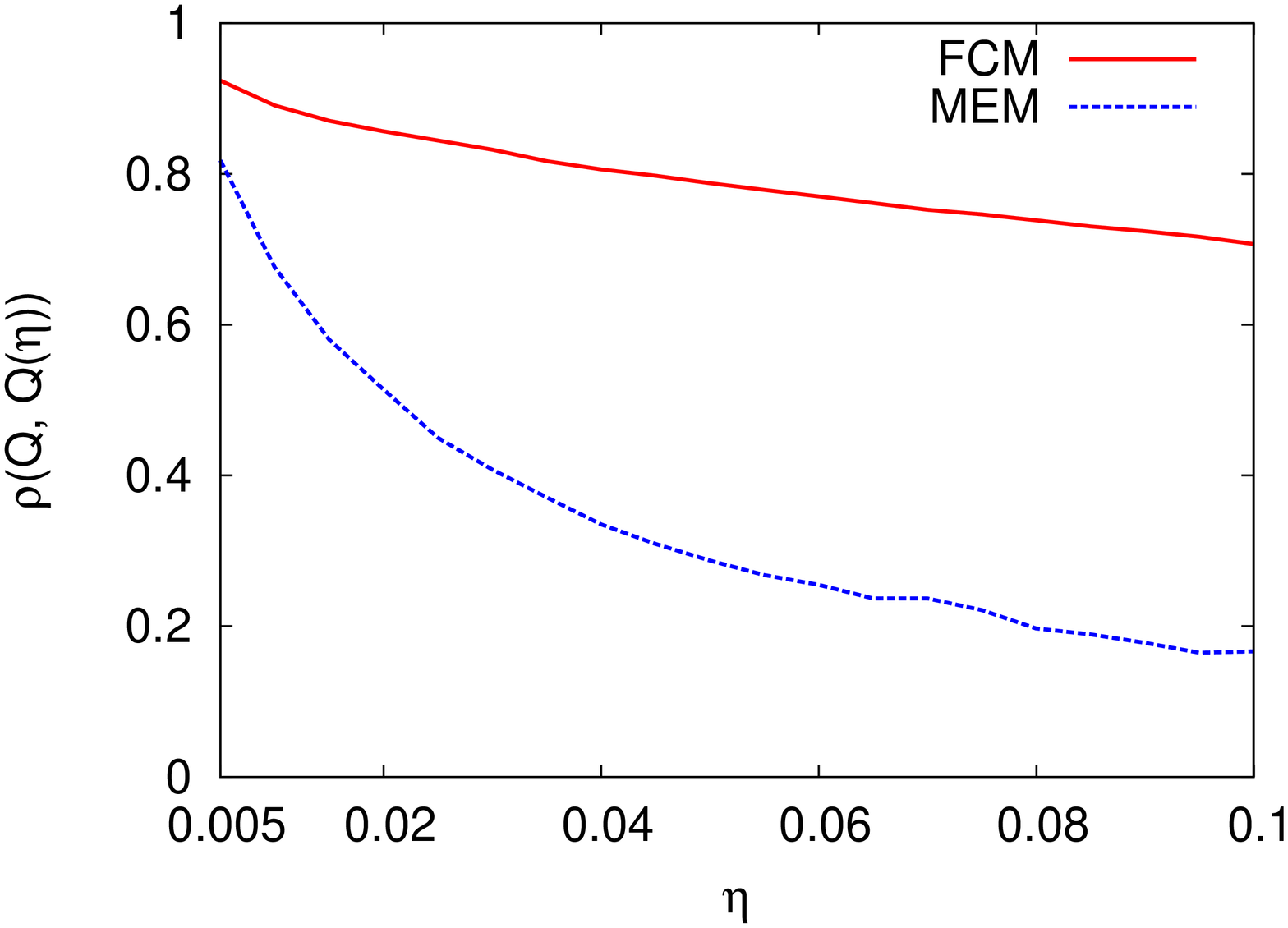}
 \caption{Robustness against noise of the rankings as a function of the fraction $\eta$ of reverted bits in the matrix $M$ (year 1996).
 Robustness is
measured by the Spearman's correlation between the rankings before and after the inversion.
 \label{fig:stability}}
\end{figure}

\subsection{Revealing the nested structure of country-product matrices}

The MEM has been introduced in section \ref{algorithms} as a minimal metric based on the same assumptions of the fitness-complexity metric.
In this section, we explore its behavior on real data and compare its rankings with those produced by the FCM.
In real data, the fitness-complexity metric has been used to
reveal the nested structure of a given network.
This is achieved by ordering the rows and the columns of the matrix $\mathcal{M}$ according to their ranking by the metric
\cite{tacchella2012new,dominguez2015ranking}.
In particular, the fitness-complexity metric outperforms other existing network centrality metrics and standard nestedness calculators
in packing nested matrices \cite{dominguez2015ranking}.
Here, we use the NBER-UN international trade data to compare the matrices produced by the FCM and the MEM;
we refer to Appendix C for a detailed description of the dataset. 
We show here results for the year $1996$; 
results for the different years are
in qualitative agreement.
We first observe that the rankings of countries by the two metrics are highly correlated ($\rho=0.994$),
and both country scores are highly correlated with country diversification [$\rho=0.963$ for the FCM, $\rho=0.955$ for the MEM].
With respect to the matrix produced by the FCM, the matrix produced by the MEM 
exhibits a sharper border between the empty and the filled parts of the matrix (see Fig. \ref{fig:matrix}).
This result suggests that the MEM could be used to produce optimally packed matrices for networks that exhibit a nested structure \cite{dominguez2015ranking}.
In agreement with the convergence criterion introduced in Section \ref{mem_results},
to obtain the results shown in the Fig. \ref{fig:matrix}, we performed $107$ and $6700$ iterations for the FCM and the MEM, respectively.
We refer to \ref{convergence} and \ref{complexity} for more details on the convergence properties of the two metrics in real and artificial data.


\subsection{Sensitivity with respect to noisy input data}

An important issue for any data-driven variable is its stability with respect to
perturbations in the system \cite{lu2011leaders, battiston2014metrics, liu2016stability}.
Following Refs. \cite{battiston2014metrics, mariani2015measuring}, 
in order to study the robustness of the rankings with respect to noise,
we randomly revert a fraction $\eta$ of bits in the binary matrix $\mathcal{M}$ and compute 
the Spearman's correlations of the scores computed before and after the reversal.
Fig. \ref{fig:stability} shows that the rankings by the FCM are more stable than the rankings by the MEM;
the gap between the two methods is particularly large for the ranking of products.
On the other hand, the different robustness of the two methods is mostly due to the region of the matrix $\mathcal{M}$
whose elements correspond to the most complex products and the least-fit countries.
This can be proved by perturbing only the submatrix $\mathcal{M}^{(bottom-right)}$ 
of $\mathcal{M}$ that contains the $M/2$ most complex products and the $N/2$ least-fit countries according to the ranking by the MEM,
and compare the outcome with that obtained when perturbing only the submatrix $\mathcal{M}^{(top-left)}$ that 
contains the $M/2$ least-complex products and the $N/2$ most fit countries.
The difference is striking: for $\mathcal{M}^{(bottom-right)}$, we find $\rho(Q,Q(0.1))=0.420$ and $\rho(Q,Q(0.1))=-0.142$ for the
FCM and the MEM, respectively;
for $\mathcal{M}^{(top-left)}$, we find $\rho(Q,Q(0.1))=0.994$ and $\rho(Q,Q(0.1))=0.999$ for the FCM and the MEM, respectively.
These findings indicate that the ranking of products by the MEM is not reliable
when the data are subject to mistakes and noise, as is the case for world trade data \cite{battiston2014metrics},
and that the major contribution to the ranking instability comes from the exports of the least-fit countries.

\section{Conclusions}

Understanding the mathematics behind a network-based ranking algorithm is crucial for its real-world applications.
This article moves the first step toward a rigorous understanding of the mathematical properties of the fitness-complexity metric
for nested networks.
We exactly computed country and product scores for perfectly nested matrices.
Our analytic findings 
are in agreement with the analytic and numeric findings of Ref. \cite{pugliese2014convergence} on the relation
between the convergence properties of the metric and the shape of the underlying nested matrix.
We stress again that while we have used the terminology of economic complexity throughout this work,
our findings hold for any network that exhibit a nested architecture.
For the application of the metric to mutualistic networks, only the meaning of the variables change:
$F$ and $Q$ represent active species importance and passive species vulnerability, respectively \cite{dominguez2015ranking}.

In this work, we have also introduced and studied the MEM, a novel variant of the FCM that is simpler to be analytically treated.
Our findings on real data indicate that the MEM can order rows and columns of nested matrices even better than the FCM.
The high correlation between the country scores obtained with the FCM and the MEM suggests that 
the MEM and the FCM are similarly informative about the competitiveness of a country in international trade
and its future growth potential \cite{cristelli2015heterogeneous}.
On the other hand, the rankings of products by the MEM turn out to be much less stable under a random perturbation in the country-product binary matrix.
To conclude, while the MEM can produce more packed nested matrices with respect to those produced by the FCM, its
ranking of products is reliable only for high-quality data.

\appendix

\section{Proof of Eq. \eqref{solution_minimal}}
\label{mem_proof}

We assume that Eq. \eqref{solution_minimal} holds for $i=k$, and show that this assumption implies that it holds also for $i=k+1$.
In formulas, we assume that in the limit $n\to\infty$
\begin{equation}
 \frac{F_{k}^{(n)}}{F_{k+1}^{(n)}}=1-\frac{\Delta_{k+1}}{H_{k+1}},
 \label{hyp}
\end{equation}
where $H_{k+1}=\max_{j\in[1,k+1]}{\{\Delta_{j}\}}$, and we want to prove that Eq. \eqref{hyp} implies
\begin{equation}
 \frac{F_{k+1}^{(n)}}{F_{k+2}^{(n)}}=1-\frac{\Delta_{k+2}}{H_{k+2}}.
 \label{thesis}
\end{equation}
Using Eq. \eqref{fit_iter}, we obtain
\begin{equation}
 \frac{F^{(n+1)}_{k+2}}{F^{(n+1)}_{k+1}}=\frac{F^{(n+1)}_{k+1}+F^{(n)}_{k+2}\,\Delta_{k+2}}{F_{k+1}^{(n+1)}}=1+\frac{F^{(n)}_{k+2}\,\Delta_{k+2}}{F_{k+1}^{(n+1)}}.
\label{proof1}
 \end{equation}
We want to express the denominator $F_{k+1}^{(n+1)}$ in terms of $F_{k+1}^{(n)}$ in order to transform this equation
into a recurrence relation for $\frac{F^{(n)}_{k+2}}{F^{(n)}_{k+1}}$.
To do this, we use Eq. \eqref{fit_iter} and obtain
\begin{equation}
 F_{k+1}^{(n+1)}=  F_{k}^{(n+1)} + F_{k+1}^{(n)}\,\Delta_{k+1}=F_{k+1}^{(n)}\,
 \Delta_{k+1}\,\Biggl(1+\frac{F_{k}^{(n+1)}}{\Delta_{k+1}\,F_{k+1}^{(n)}} \Biggr)
 =F_{k+1}^{(n)}\, \Delta_{k+1}\,\Biggl(1+\frac{F_{k}^{(n+1)}}{F_{k+1}^{(n+1)}-F_{k}^{(n+1)}} \Biggr).
\end{equation}
We now use the hypothesis \eqref{hyp}:
\begin{equation}
\begin{split}
 F_{k+1}^{(n+1)}&=F_{k+1}^{(n)}\, \Delta_{k+1}\,\Biggl(1+\frac{F_{k}^{(n+1)}/F_{k+1}^{(n+1)}}{1-F_{k}^{(n+1)}/F_{k+1}^{(n+1)}}\Biggr)
 =F_{k+1}^{(n)}\, \Delta_{k+1}\,\Biggl(1+\frac{1-\Delta_{k+1}/H_{k+1}}{\Delta_{k+1}/H_{k+1}}\Biggr)\\
 &=F_{k+1}^{(n)}\, \Delta_{k+1}\,\Biggl(1+ \frac{H_{k+1}-\Delta_{k+1}}{\Delta_{k+1}}\Biggr)=F_{k+1}^{(n)}\,H_{k+1}.
 \end{split}
 \label{intermediate}
\end{equation}
Plugging Eq. \eqref{intermediate} into Eq. \eqref{proof1} we get
\begin{equation}
 \frac{F^{(n+1)}_{k+2}}{F^{(n+1)}_{k+1}}=1+\frac{\Delta_{k+2}}{H_{k+1}}\, \frac{F^{(n)}_{k+2}}{F_{k+1}^{(n)}},
\label{proof2}
 \end{equation}
Eq. \eqref{proof2} is a recurrence equation for $\frac{F^{(n+1)}_{k+2}}{F^{(n+1)}_{k+1}}$.
We distinguish two cases:
\begin{itemize}
 \item If $\frac{\Delta_{k+2}}{H_{k+1}}\geq 1$, then $\lim_{n\to\infty}F^{(n)}_{k+2}/F^{(n)}_{k+1}=\infty$ and $\lim_{n\to\infty}F^{(n)}_{k+1}/F^{(n)}_{k+2}=0$.
 In this case $H_{k+2}=\Delta_{k+2}$ by definition and Eq. \eqref{thesis} (i.e., the thesis) is satisfied.
 \item If $\frac{\Delta_{k+2}}{H_{k+1}}<1$, then we can find the stationary point $\overline{x}$ of equation \eqref{proof2} by
 posing $\overline{x}=\frac{F^{(n+1)}_{k+2}}{F^{(n+1)}_{k+1}}=\frac{F^{(n)}_{k+2}}{F^{(n)}_{k+1}}$. We notice that
 in this case $H_{k+2}=H_{k+1}$ and, as a result, we obtain Eq. \eqref{thesis}.
\end{itemize}
This proves the thesis.

\section{Proof of Theorem 2}

We denote by $(\mathbf{x},\mathbf{y})$ the pair of vectors that solve the Eqs. \eqref{equivalent1}-\eqref{equivalent4} in the limit $n\to\infty$,
which read
\begin{equation}
 x_{i}=\frac{\delta_{i}\,y_{i}}{\delta_{i}\,y_{i}+\delta_{i+1}\,(1-x_{i-1})}, \,\,\, (i=2,\dots,m-1)
 \label{stationary1}
\end{equation}
\begin{equation}
 x_{1}=\frac{\delta_{1}\,y_{1}}{\delta_{1}\,y_{1}+\delta_{2}},
 \label{stationary2}
\end{equation}
\begin{equation}
 y_{i}=\frac{\epsilon_{i+1}x_{i}}{\epsilon_{i+1}x_{i}+\epsilon_{i}(1-y_{i+1})}, \,\,\, (i=1,\dots,m-2)
 \label{stationary3}
\end{equation}
\begin{equation}
 y_{m-1}=\frac{\epsilon_{m}x_{m-1}}{\epsilon_{m}x_{m-1}+\epsilon_{m-1}}.
 \label{stationary4}
\end{equation}
To prove that $(\mathbf{x},\mathbf{y})=(\mathbf{a},\mathbf{b})$ is a solution of the Eqs. \eqref{stationary1}-\eqref{stationary4}, it is sufficient to check that
an identity is obtained when
replacing $x_{i}$ and $y_{i}$ with $a_{i}$ and $b_{i}$ in the Eqs. \eqref{stationary1}-\eqref{stationary4}.
If $e_m\,d_i>e_i\,d_m$, we can easily use mathematical induction 
to prove that $x_i^{(n)}>a_i$ and $y_i^{(n)}>b_i$ for all $i=1,\dots,m-1$ and for all $n\geq 0$. The proof is similar to the proof of Lemma 1.
We are then interested only in solutions of the Eqs. \eqref{stationary1}-\eqref{stationary4} that satisfies
\begin{equation}
a_i\leq x_i<1
\label{inequality1}
\end{equation}
and
\begin{equation}
b_i\leq y_i< \frac{e_m-e_i}{e_m-e_{i-1}};
\label{inequality2}
\end{equation}
solutions that do not satisfy conditions \eqref{inequality1} and \eqref{inequality2} cannot be reached through the iterative
process defined by Eqs. \eqref{equivalent1}-\eqref{equivalent4}, and will not be considered in the following.
In the following, always imply conditions \eqref{inequality1} and \eqref{inequality2} for the studied solutions.
To show that the solution of Eqs. \eqref{stationary1}-\eqref{stationary4} is 
unique, we use a reductio ad absurdum:
we assume that a different solution $\mathbf{x}=\mathbf{\tilde{a}}$ and $\mathbf{y}=\mathbf{\tilde{b}}$ 
exists, and show that this assumption leads to an absurd result.
Before doing this, we state two useful properties of the solutions of Eqs. \eqref{stationary1}-\eqref{stationary4}.

\begin{property}
For a solution $(\mathbf{x},\mathbf{y})$ of Eqs. \eqref{stationary1}-\eqref{stationary4}, if there exist an integer $j$ such that $x_{j}>a_{j}$ or $y_{j}>b_{j}$,
then $x_{i}>a_{i}$ and $y_{i}>b_{i}$ for all $i=1,\dots,m$.
\label{property1}
\end{property}

\begin{proof}
 Suppose that $x_{j}>a_{j}$ for a certain component $j$.
 Then 
 \begin{equation}
  y_{j}=\frac{\epsilon_{j+1}\,x_{j}}{\epsilon_{j+1}\,x_{j}+\epsilon_{j}\,(1-y_{j+1})}>\frac{\epsilon_{j+1}\,a_{j}}{\epsilon_{j+1}\,a_{j}+\epsilon_{j}\,(1-y_{j+1})}      
 \geq \frac{\epsilon_{j+1}\,a_{j}}{\epsilon_{j+1}\,a_{j}+\epsilon_{j}\,(1-b_{j+1})}=b_{j}.
 \end{equation}
In a similar way, one can use Eq. \eqref{stationary1} to prove the thesis for all $i>j$, and Eq. \eqref{stationary3} to prove the thesis for all $i<j$.
\end{proof}

In order to have a solution $(\mathbf{\tilde{a}},\mathbf{\tilde{b}})$ such that 
$\tilde{a}\neq a$ and $\tilde{b}\neq b$, there must exist at least one component $j$ such that $\tilde{a}_{j}\neq a_{j}$
or $\tilde{b}_{j}\neq b_{j}$;
from the inequalities \eqref{inequality1}-\eqref{inequality2} and Property \ref{property1}, we also have $a_{i}<\tilde{a}_{i}<1$ or
$b_{i}<\tilde{b}_{i}<(e_m-e_i)/(e_m-e_{i-1})$ for $i=1,\dots,m$.

\begin{property}
If $y_{i}>0\,\forall i=2,3,...,m-1$, for each solution $(\mathbf{x},\mathbf{y})$ of the Eqs. \eqref{stationary1}-\eqref{stationary4},
the value of $y_{m-1}$ uniquely determines the values of all the other components $\{x_{i}\}_{i=1}^{m-1}$
and $\{y_{i}\}_{i=1}^{m-2}$ of the solution. On the other hand, $y_{m-1}$ is uniquely determined by the other components $\{x_{i}\}_{i=1}^{m-1}$
and $\{y_{i}\}_{i=1}^{m-2}$ of the solution.
\label{property2}
\end{property}

\begin{proof}
The former statement of the Property follows from the fact that if $y_{i}\neq 0\, \forall i=2,...,m-1$ and we know the last component $y_{m-1}$ of the solution,
we can then compute all the other components of the solution and they uniquely depend on $y_{m-1}$.
Suppose indeed that we know the value of $y_{m-1}$. We can then invert Eq. \eqref{stationary4} and compute
$x_{m-1}=\epsilon_{m-1}y_{m-1}/\epsilon_m(1-y_{m-1})$,
and then plug the obtained $x_{m-1}$ value into Eq. \eqref{stationary1} to compute $x_{m-2}$, and then plug the obtained $x_{m-2}$ into
Eq. \eqref{stationary3} to compute $y_{m-2}$, and so on.
The latter statement follows from Property \ref{property1} (or equivalently, from 
the invertibility of all the relations involved in Eqs. \eqref{stationary1}-\eqref{stationary4}).
\end{proof}

As a consequence of this property, proving that the solution $(\mathbf{x},\mathbf{y})=(\mathbf{a},\mathbf{b})$ 
is unique is equivalent to proving that for a solution,
the only acceptable value of $y_{m-1}$ is $y_{m-1}=b_{m-1}$.

We will now prove the theorem in two steps:
\begin{enumerate}
 \item We transform Eqs. \eqref{stationary1}-\eqref{stationary4} 
 into a set of equations, hereafter referred to as the transformed equations.
\item We use a reductio ad absurdum and assume that there exists a solution $y=\tilde{b}$ of the original equations such that $\tilde{b}_{N-1}>b_{N-1}$.
 We use then the transformed equations to show that the solution $y=\tilde{b}$ cannot be a solution of the original equations, which proves the thesis.
 \end{enumerate}

\subsection{Step 1: Deriving a set of transformed equations}

First, we merge the equations \eqref{stationary1}-\eqref{stationary4} into two equations
\begin{equation}
 x_{i}=\frac{y_{i}}{y_{i}+\frac{\delta_{i+1}}{\delta_{i}}\,(1-x_{i-1})},\,\,\,(i=1,\dots m-1)
 \label{stationary1_new}
\end{equation}
\begin{equation}
 y_{i}=\frac{x_{i}}{x_{i}+\frac{\epsilon_i}{\epsilon_{i+1}}(1-y_{i+1})},\,\,\,(i=1,\dots m-1)
 \label{stationary2_new}
\end{equation}
with $x_{0}=0$ and $y_{m}=0$.
Consider a generic solution $(\mathbf{\bar{x}},\mathbf{\bar{y}})$ of Eqs. \eqref{stationary1_new}-\eqref{stationary2_new}.
Instead of the variables $\{x_{1},\dots, x_{m-1}\}$  and $\{y_{1},\dots, y_{m-1}\}$,
we consider the transformed variables $\{x'_{1},\dots, x'_{m}\}$  and $\{y'_{1},\dots, y'_{m}\}$
defined by the transformation 
\begin{equation} 
\begin{split} 
x'_{i}&=x_{i-1} \,\,\,\,\,\text{for} \,\,\,i=2,\dots,m;\\
y'_{i}&=y_{i-1} \,\,\,\,\,\text{for} \,\,\,i=2,\dots,m.
\end{split}
\label{transformation}
\end{equation}
We consider the transformed equations
\begin{equation}
 x'_{i}=\frac{y'_{i}}{y'_{i}+\frac{\delta'_{i+1}}{\delta'_{i}}\,(1-x'_{i-1})},\,\,\,(i=1,\dots m),
 \label{transformed1}
\end{equation}
\begin{equation}
 y'_{i}=\frac{x'_{i}}{x'_{i}+\frac{\epsilon'_i}{\epsilon'_{i+1}}(1-y'_{i+1})},\,\,\,(i=1,\dots m),
 \label{transformed2}
\end{equation}
with $x'_{0}=0=y'_{m+1}=0$, $\delta'_{i}=\delta_{i-1}$ and $\epsilon'_{i}=\epsilon_{i-1}$ for $i=2,\dots,m+1$.
In the tranformed equations, $x'_{1}$ and $y'_{1}$ are new variables; 
for a solution $(\mathbf{\bar{x}},\mathbf{\bar{y}})$ of Eqs, \eqref{stationary1_new}-\eqref{stationary2_new},
the pair of transformed vectors  $(\mathbf{\bar{x}'},\mathbf{\bar{y}'})$ satisfies the following set of transformed equations
only if $\bar{x}'_1=\bar{y}'_1=0$ for a solution $(\mathbf{\bar{x}},\mathbf{\bar{y}})$.
The values of $\delta'_{1}$ and $\epsilon'_{1}$ only affect the values of $x'_{1}$ and $y'_{1}$, which must be equal to zero for 
the transformed $(\mathbf{\bar{x}'},\mathbf{\bar{y}'})$ of a solution $(\mathbf{\bar{x}},\mathbf{\bar{y}})$ of the original equations.
This allows us to let $\delta'_{1}$ and $\epsilon'_{1}$ be arbitrary parameters in the Eqs. \eqref{transformed1}-\eqref{transformed2}.
Eqs. \eqref{transformed1}-\eqref{transformed2} have the same form of Eqs. \eqref{stationary1_new}-\eqref{stationary2_new}.
It is possible to show by substitution that a possible solution of Eqs. \eqref{transformed1}-\eqref{transformed2} is
\begin{equation}
 \bar{x}'_{i}=\frac{e'_{m+1}\,d'_{i}-e'_i\,d'_{m+1}}{e'_{m+1}\,d'_{i+1}-e'_i\,d'_{m+1}},
 \label{a_transf}
 \end{equation}
 \begin{equation}
   \bar{y}'_{i}=\frac{e'_{m+1}\,d'_{i}-e'_i\,d'_{m+1}}{e'_{m+1}\,d'_{i}-e'_{i-1}\,d'_{m+1}},
   \label{b_transf}
\end{equation}
where $d'_{i}=\delta'_{1}+\sum_{j=2}^{i}\delta'_{j}$ and $e'_{i}=\epsilon'_{1}+\sum_{j=2}^{i}\epsilon'_{j}$.
The $m$-th component of this solution is 
\begin{equation}
\bar{y}'_m=\frac{e'_{m+1}d'_{m}-e'_md'_{m+1}}{e'_{m+1}d'_{m}-e'_{m-1}d'_{m+1}}=\frac{(\epsilon'_{1}+e_m)(\delta'_1+d_{m-1})-(\epsilon'_{1}+e_{m-1})(\delta'_1+d_{m})}{(\epsilon'_{1}+e_m)(\delta'_1+d_{m-1})-(\epsilon'_{1}+e_{m-2})(\delta'_1+d_{m})}.
\end{equation}
We are interested in the 
solutions $(\bar{x}',\bar{y}')$ of Eqs. \eqref{transformed1}-\eqref{transformed2} such that $(\mathbf{\bar{x}},\mathbf{\bar{y}})$
is solution of Eqs. \eqref{stationary1_new}-\eqref{stationary2_new}, where the transformation $(\mathbf{\bar{x}'},\mathbf{\bar{y}'})\to(\mathbf{\bar{x}},\mathbf{\bar{y}})$
is given by Eq. \eqref{transformation}.
We can then pose $\bar{y}'_m=\bar{y}_{m-1}$ and $\epsilon'_{1}=1$, and obtain
\begin{equation}
 \delta'_1=\frac{(1+e_m)\,\delta_m\,(1-\bar{y}_{m-1})}{\epsilon_m-(\epsilon_m+\epsilon_{m-1})\,\bar{y}_{m-1}}-d_m.
 \label{deltaprime}
\end{equation}

\subsection{Step 2: Reductio ad absurdum}

Up to now, we have proven for a solution $(\mathbf{\bar{x}'},\mathbf{\bar{y}'})$ of the Eqs. \eqref{transformed1}-\eqref{transformed2} 
in the form given by Eqs. \eqref{a_transf}-\eqref{b_transf}, the pair of vectors $(\mathbf{\bar{x}},\mathbf{\bar{y}})$ obtained by the transformation \eqref{transformation}
is a solution of 
Eqs. \eqref{stationary1_new}-\eqref{stationary2_new} if and only if $\bar{x}'_{1}=\bar{y}'_{1}=0$,
$\delta'_{1}$ satisfies Eq. \eqref{deltaprime} and $\epsilon'_{1}=1$.
We will now show that if we consider a solution of the transformed equations such that $\bar{y}'_{m}=\tilde{b}'_{m}=\tilde{b}_{m-1}>b_{m-1}$
and assume that $(\mathbf{\tilde{a}},\mathbf{\tilde{b}})$ is a solution of the original equations,
then the first component $\tilde{a}'_{1}$ of the solution is different from zero, which is absurd.
As a consequence, $(\mathbf{\bar{x}},\mathbf{\bar{y}})=(\mathbf{a},\mathbf{b})$ is the only solution
of Eqs. \eqref{stationary1_new}-\eqref{stationary2_new}.

\begin{proof}
We assume that there exists a solution $\bar{y}_{m-1}=\tilde{b}_{m-1}>b_{m-1}$; from Property \ref{property2}, all its components $(\mathbf{\tilde{a}},\mathbf{\tilde{b}})$
are uniquely determined by $\tilde{b}_{m-1}$.
Using the solution \eqref{a_transf}-\eqref{b_transf} of the transformed equations,
from $\tilde{b}_{m-1}>b_{m-1}$ and Eq. \eqref{deltaprime} it follows that 
$\delta'_{1}>d_{m}/e_{m}$.
To prove the thesis, we start by showing that $e'_{m+1}\,d'_i>e'_i\,d'_{m+1}$. For $i=2,3,...m$, we have
\begin{equation*}
e'_{m+1}d'_i-e'_id'_{m+1}=(1+e_m)(d_{i-1}+\delta'_1)-(1+e_{i-1})(\delta'_1+d_m)
\end{equation*}
\begin{equation*}
=(1+e_m)d_{i-1}+(e_m-e_{i-1})\delta'_1-(1+e_{i-1})d_m
\end{equation*}
\begin{equation*}
>(1+e_m)\frac{e_{i-1}}{e_m}d_m+(e_m-e_{i-1})\frac{d_m}{e_m}-(1+e_{i-1})d_m=0.
\end{equation*}
For $i=1$, $e'_{m+1}d'_1=\delta'_1+e_m\delta'_1>\delta'_1+d_m=d'_{m+1}=e'_1d'_{m+1}$, which implies
 $e'_{m+1}d'_i>e'_id'_{m+1}$ for all $i=1,2,...,m$; as a consequence, $a'_{1}>0$, which is absurd.
\end{proof}

\section{The dataset}
\label{dataset}

We use the NBER-UN dataset which has been cleaned and further described in \cite{feenstra2005world}.
We take into account the same list of $N=132$ countries described in \cite{hidalgo2007product}.
For products, we used the same cleaning procedure of ref. \cite{vidmer2015prediction}: we removed
aggregate product categories and products with zero total export volume for a given year and
nonzero total export volume for the previous and the following years.
Products and countries with no entries after year 1993 have been removed as well.
After the cleaning procedure, the dataset consists of $M=723$ products.
To decide if we consider country $i$ to be an exporter of product $\alpha$ or not, we use the Revealed Comparative
Advantage (RCA) \cite{balassa1965trade} which is defined as
\begin{equation}
RCA_{i\alpha}=\frac{e_{i\alpha}}{\sum_\beta e_{j\beta}} \Bigg/ \frac{\sum_j
e_{j\alpha}}{\sum_{j\beta}e_{j\beta}},
\end{equation}
where $e_{i\alpha}$ is the volume of product $\alpha$ that country $i$ exports measured in thousands of US dollars.
RCA characterizes the relative importance of a given export volume of a product
by a country in comparison with this product's exports by all other countries.
We use the bipartite network representation introduced in \cite{hidalgo2009building}, where two kinds of nodes represent countries and products, respectively.
All country-product pairs with RCA values above a threshold value--set to $1$ here--are consequently joined by links between the corresponding nodes in
the bipartite network.

\section{Spearman's correlation coefficient $\rho$}

Given two variables $\mathbf{X}=\{X_{1},\dots,X_{n}\}$ and $Y=\{Y_{1},\dots,Y_{n}\}$,
we rank them in decreasing order and denote by $\mathbf{x}=\{x_{1},\dots,x_{n}\}$ and $\mathbf{y}=\{y_{1},\dots,y_{n}\}$
their corresponding ranking scores.
Equal scores are assigned equal ranking positions given by their average ranking position:
for instance, if the fourth and the fifth scores in the ranking are equal to each other, then they are both assigned
ranking score equal to $(4+5)/2=4.5$. 
The Spearman's correlation coefficient $\rho(\mathbf{X},\mathbf{Y})$ is then defined as the 
linear correlation coefficient between the ranking scores, which reads \cite{spearman1904proof}
\begin{equation}
 \rho(\mathbf{X},\mathbf{Y})=\frac{\sum_{i=1}^{n}(x_{i}-\overline{x})(y_{i}-\overline{y}))}{\sqrt{\sum_{i=1}^{n}(x_{i}-\overline{x})^2 \, \sum_{i=1}^{n}(y_{i}-\overline{y})^2 }},
\end{equation}
where $\overline{x}=n^{-1}\sum_{i=1}^{n}x_i$ denotes the mean of $\mathbf{x}$.

\section{Convergence of the metrics in real data}
\label{convergence}

 \begin{figure}[h]
 \centering
 \includegraphics[height=0.45\columnwidth,  angle=270]{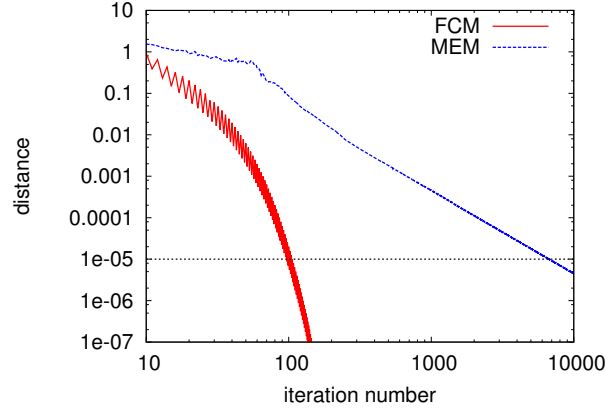}
 \caption{$d^{(n)}$ as a function of the iteration number for the country-product matrix represented in Fig. \ref{fig:matrix}.
 \label{fig:convergence}}
\end{figure}

In a perfectly nested matrix, the score ratios converge to a finite value both for the MEM (Eq. \eqref{solution_minimal}) and for the 
FCM (Eq. \eqref{solution_fitness}).
While real matrices are not perfectly nested,
one can conjecture that if the matrix is not too sparse,
the convergence behavior of a real matrix will be similar.
Motivated by this assumption, we define a convergence criterion based on the score ratio,
and decide to halt iterations when the following criterion is satisfied
\begin{equation}
 d^{(n)}=\sum_{i=1}^{N} \Bigg| \frac{F_{i}^{(n)}}{F_{i+1}^{(n)}}-\frac{F_{i}^{(n+1)}}{F_{i+1}^{(n+1)}}\Bigg|<\epsilon=10^{-5}.
 \label{convergence_criterion}
\end{equation}
For the country-product matrix shown in Fig. \ref{fig:matrix}, 
condition \eqref{convergence_criterion} is satisfied after $107$ and $6700$ iterations for the FCM and the MEM, respectively
(see Fig. \ref{fig:convergence}). 
For the FCM, we find that no fitness ratios converge to zero.
This is in agreement with our analytic results (see condition \eqref{nocrossing1}) and with the results of ref. \cite{pugliese2014convergence},
since the diagonal of the matrix never crosses the empty region of the matrix, as shown in the left panel of Fig. \ref{fig:matrix}.
For the MEM, we find that three fitness ratios converge to 
zero\footnote{We checked that the MEM fitness scores that were converging to zero were still larger than zero after $6700$ iterations.},
which slows down the convergence of the metric.




\section{Dependence of the convergence iteration of the metrics on network size}
\label{complexity}

 \begin{figure}[h]
 \centering
 \includegraphics[height=0.32\columnwidth,  angle=270]{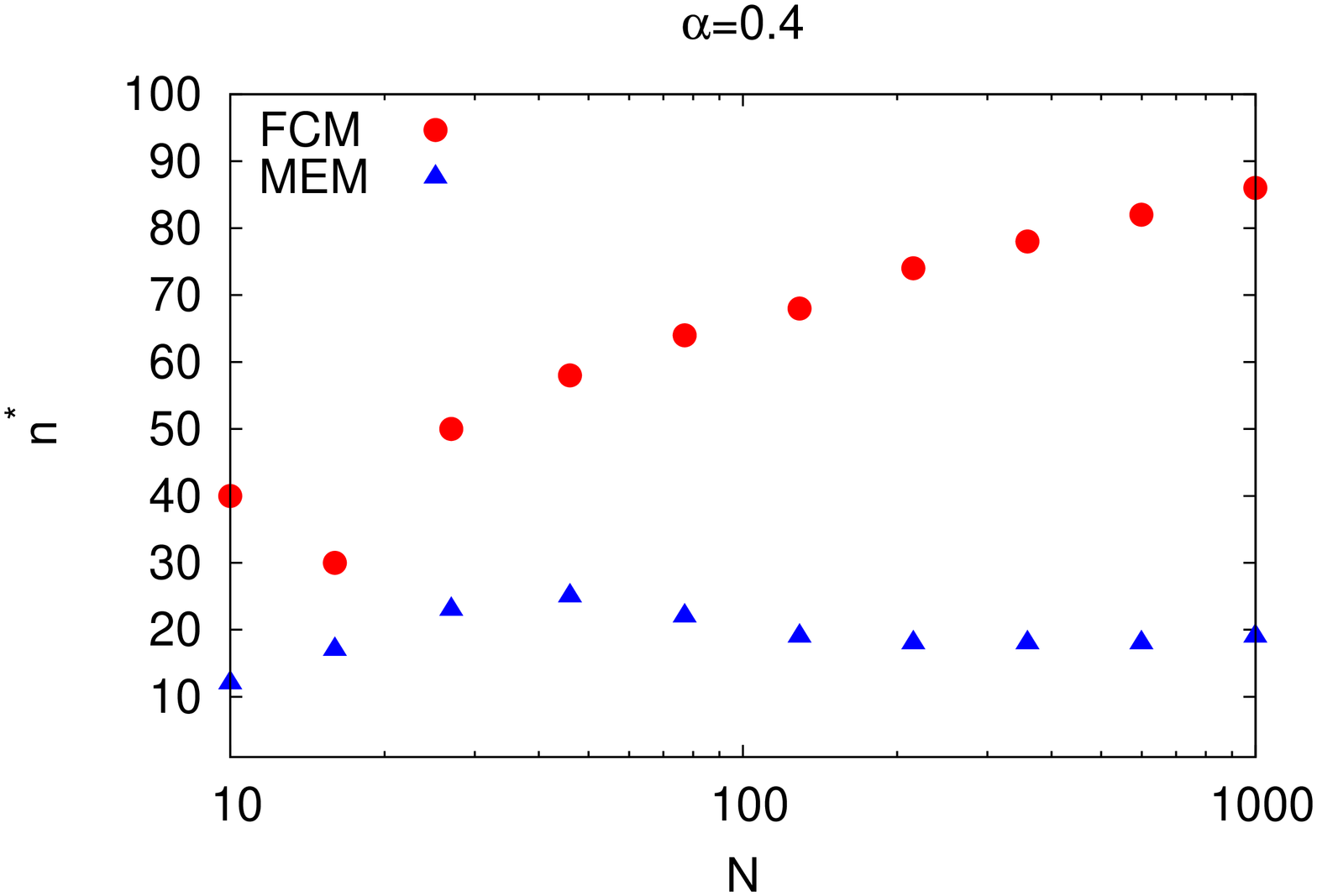}\includegraphics[height=0.32\columnwidth,  angle=270]{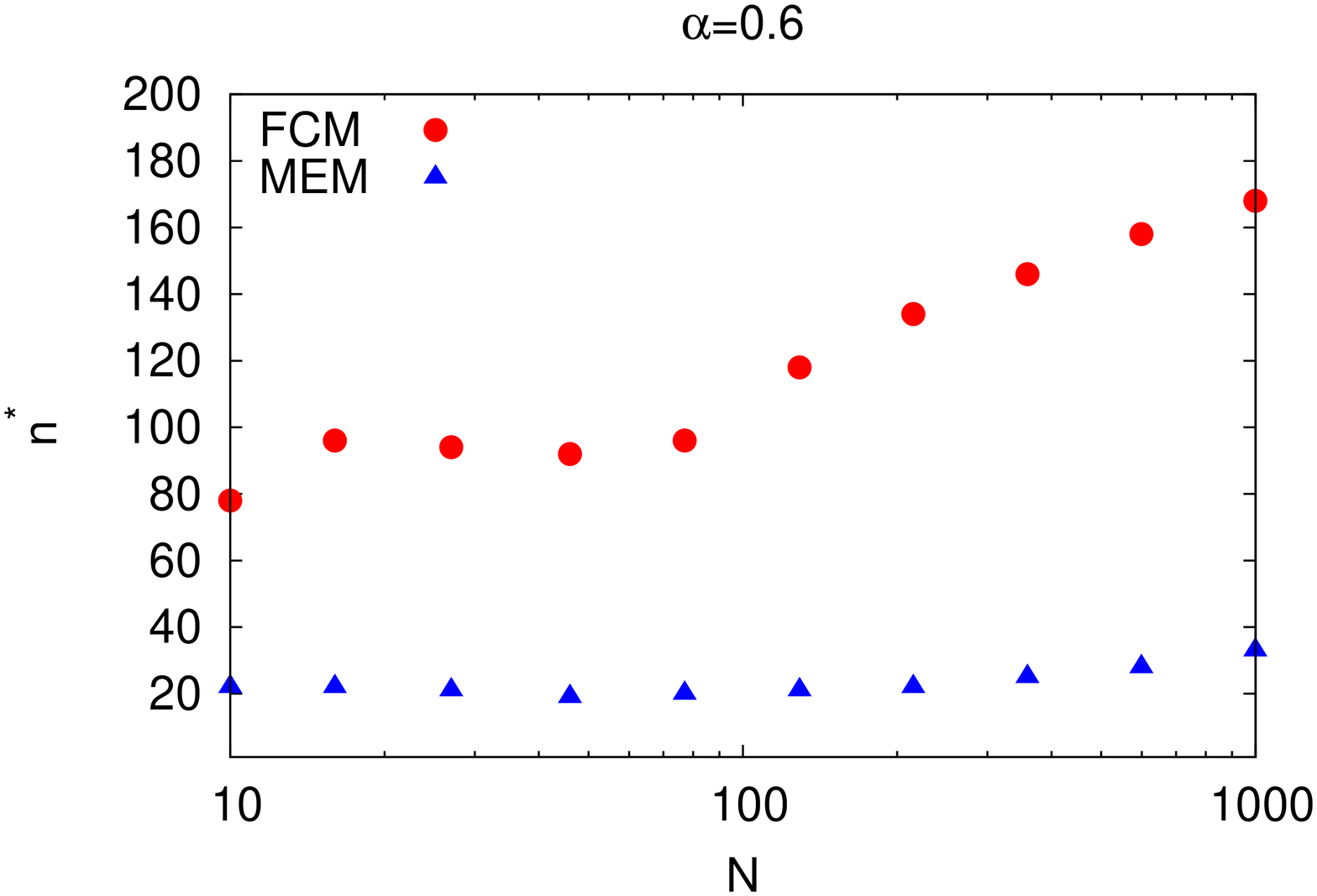}\includegraphics[height=0.32\columnwidth,  angle=270]{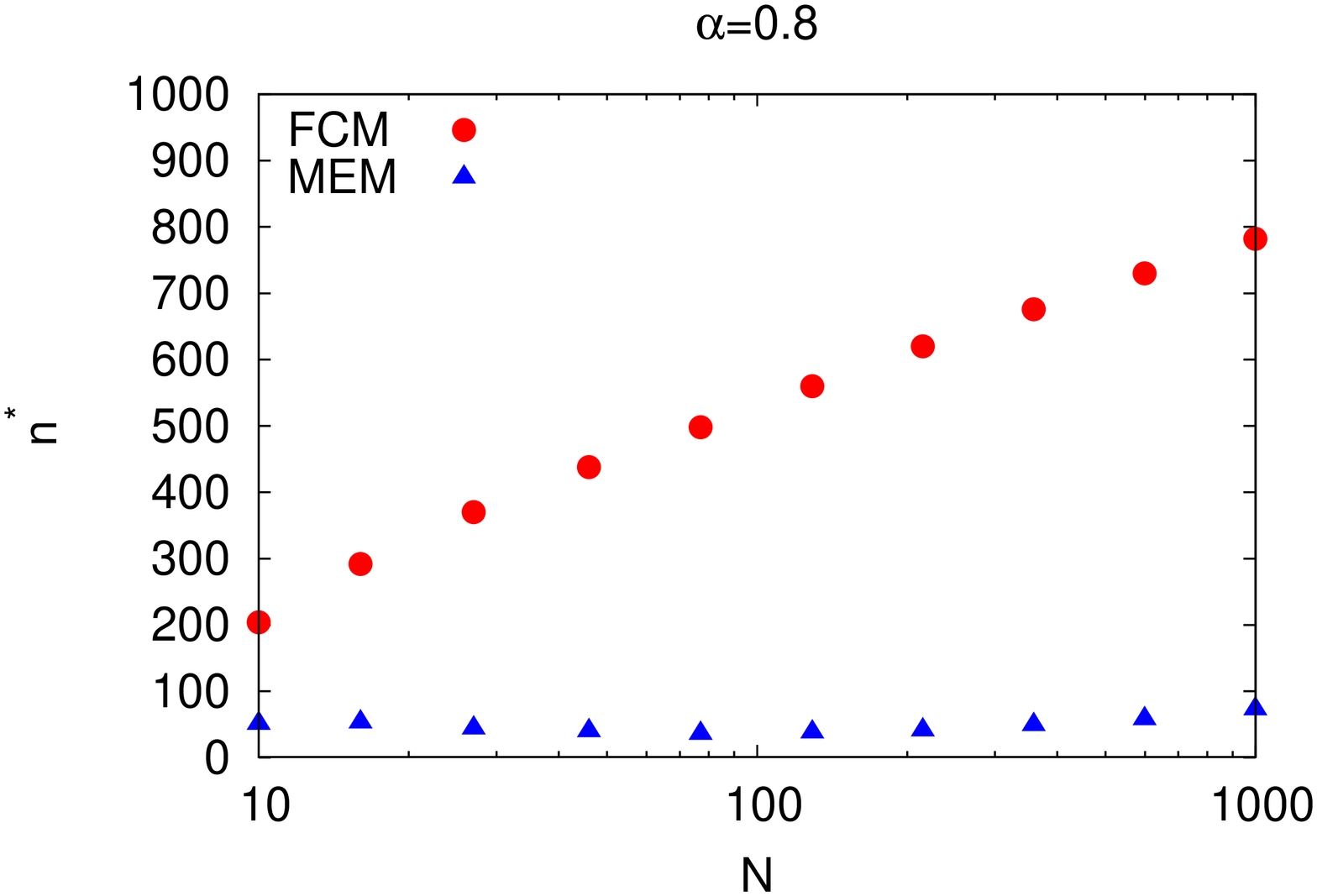}
 \caption{Convergence iteration $n^{*}$ as a function of size $N$ for artificial matrices generated 
 according to Model A described in \ref{complexity}.
 \label{fig:modelAresults}}
\end{figure}

 \begin{figure}[h]
 \centering
 \includegraphics[height=0.25\columnwidth,  angle=0]{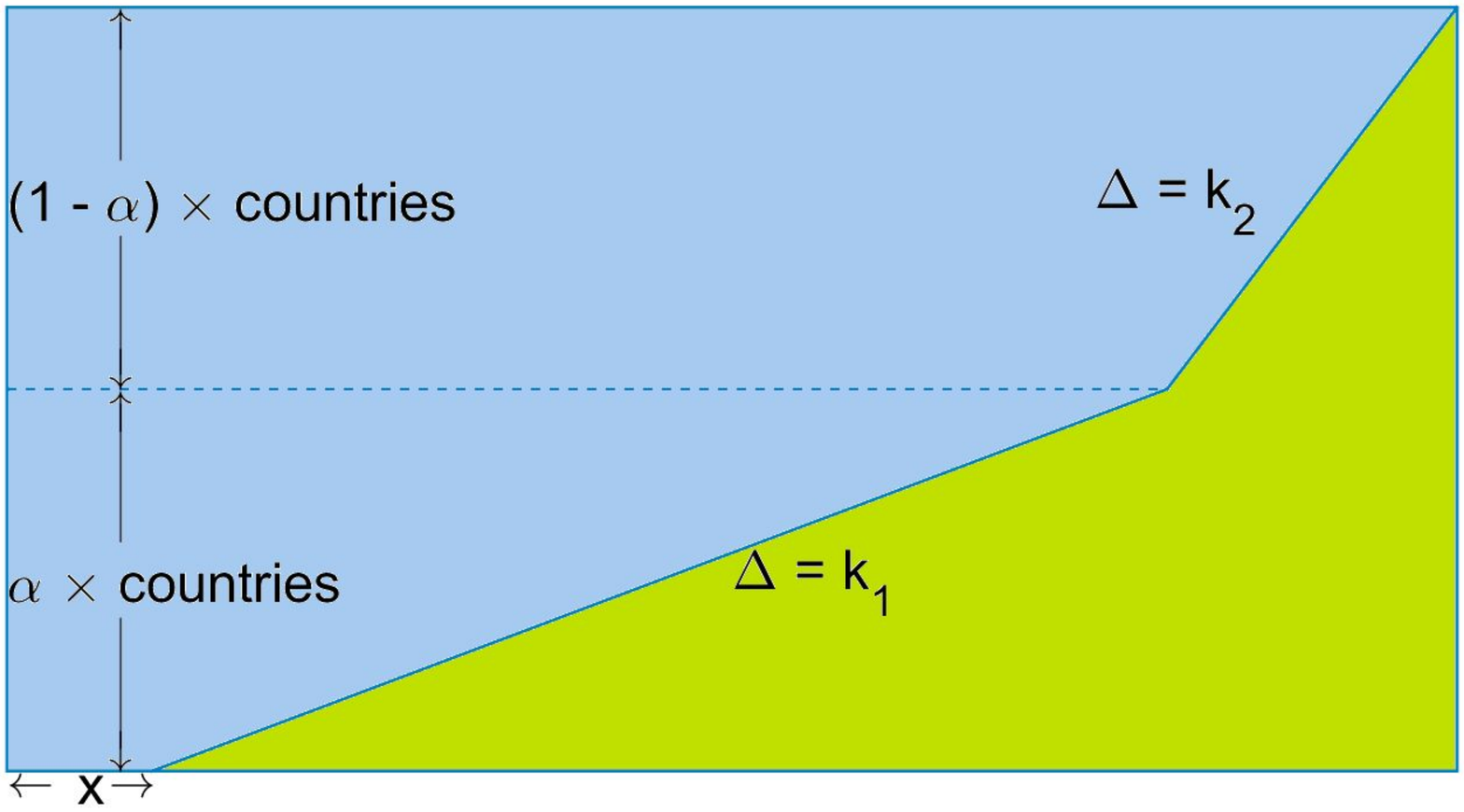}
 \caption{An illustration of model B described in \ref{complexity}.
 \label{fig:modelB}}
\end{figure}

 \begin{figure}[h]
 \centering
 \includegraphics[height=0.35\columnwidth,  angle=270]{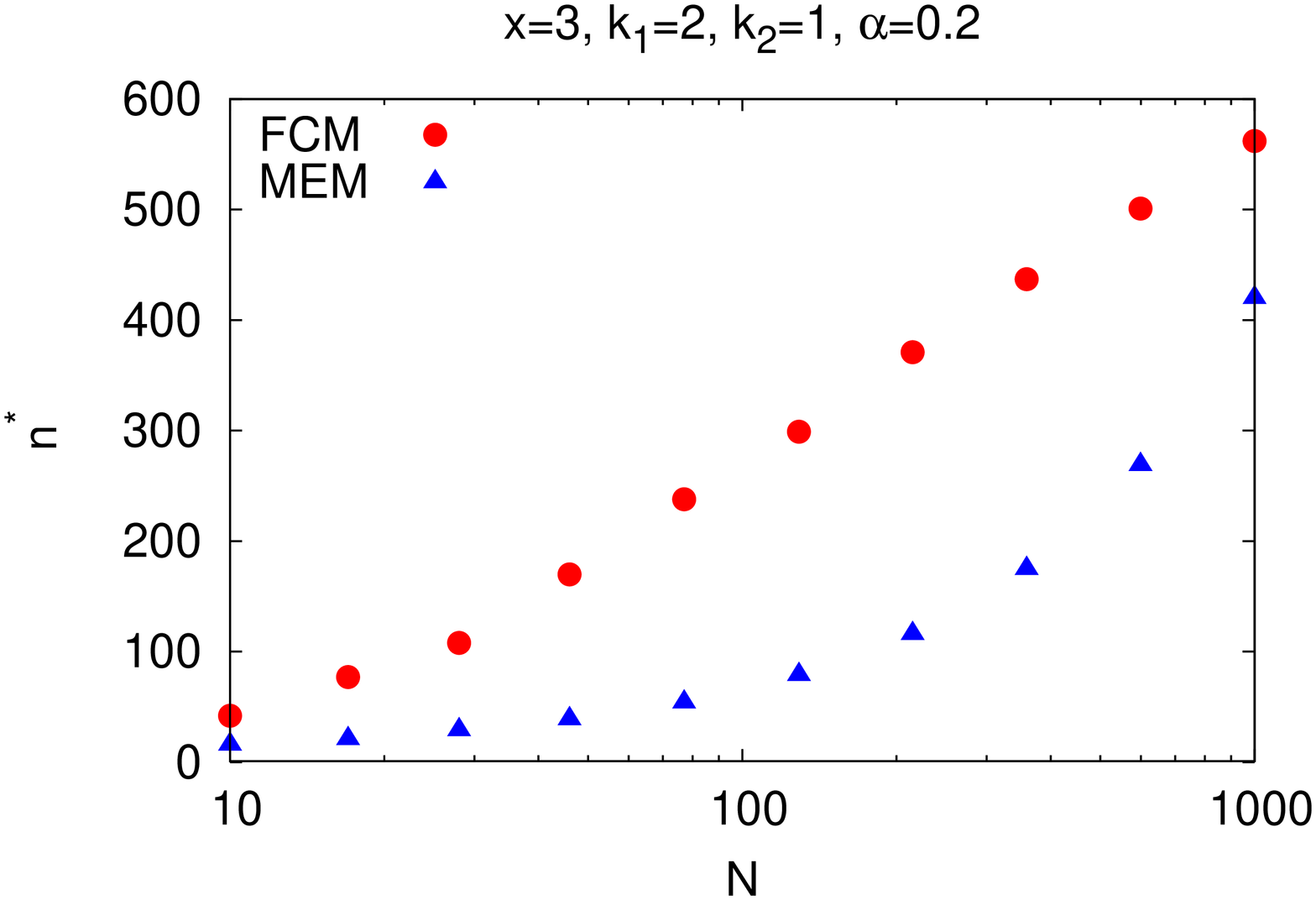}\includegraphics[height=0.35\columnwidth,  angle=270]{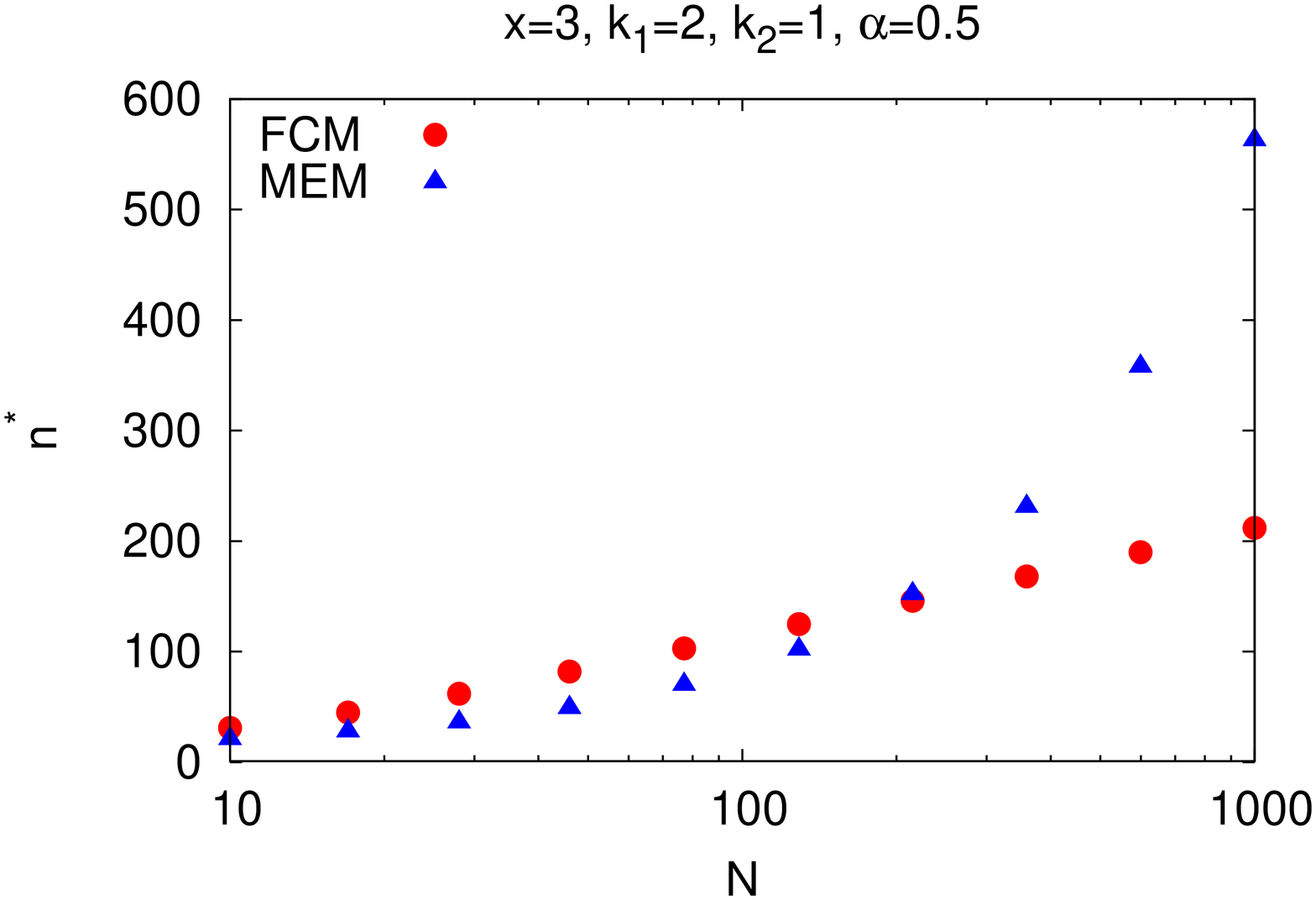}
 \includegraphics[height=0.35\columnwidth,  angle=270]{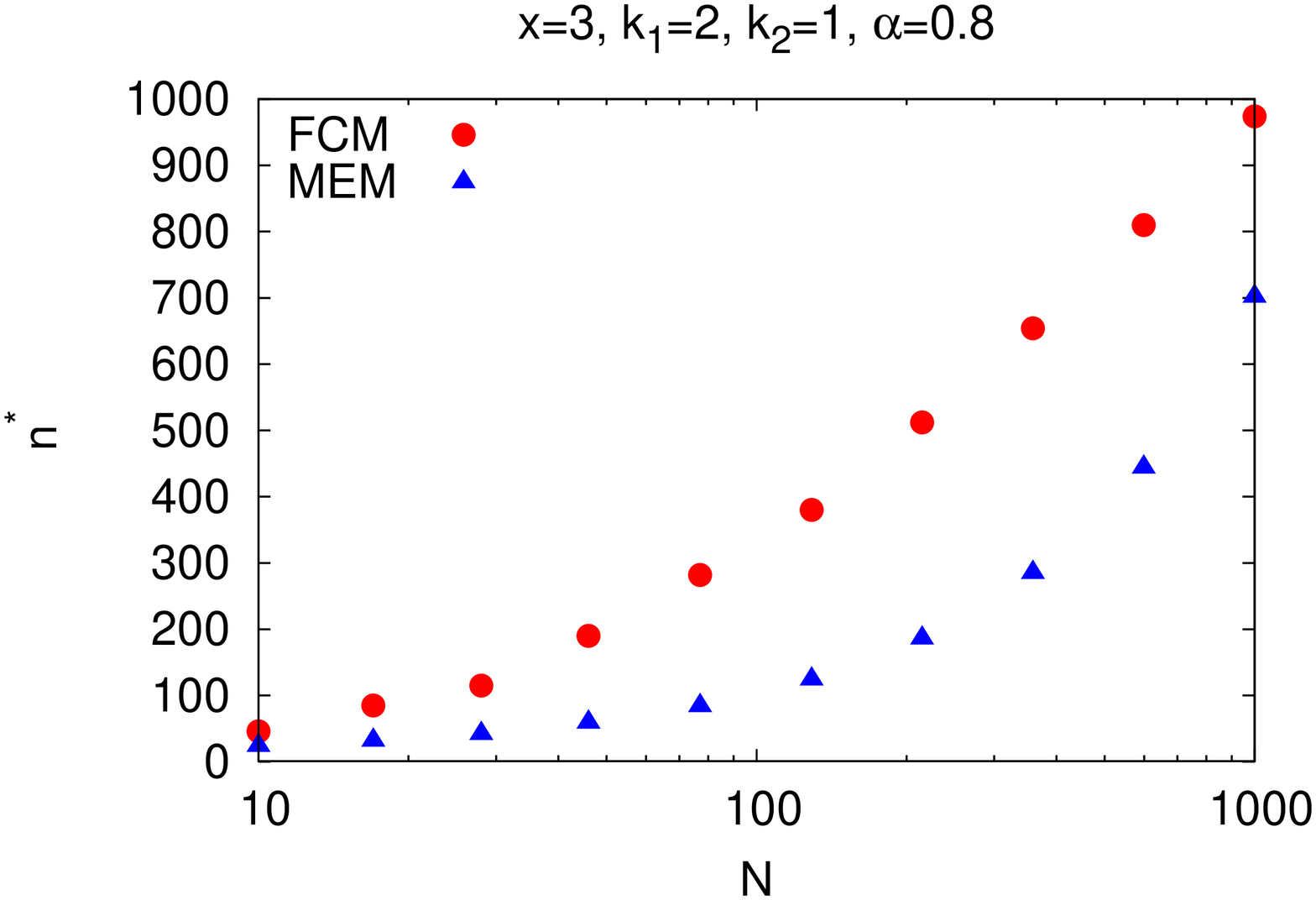}\includegraphics[height=0.35\columnwidth,  angle=270]{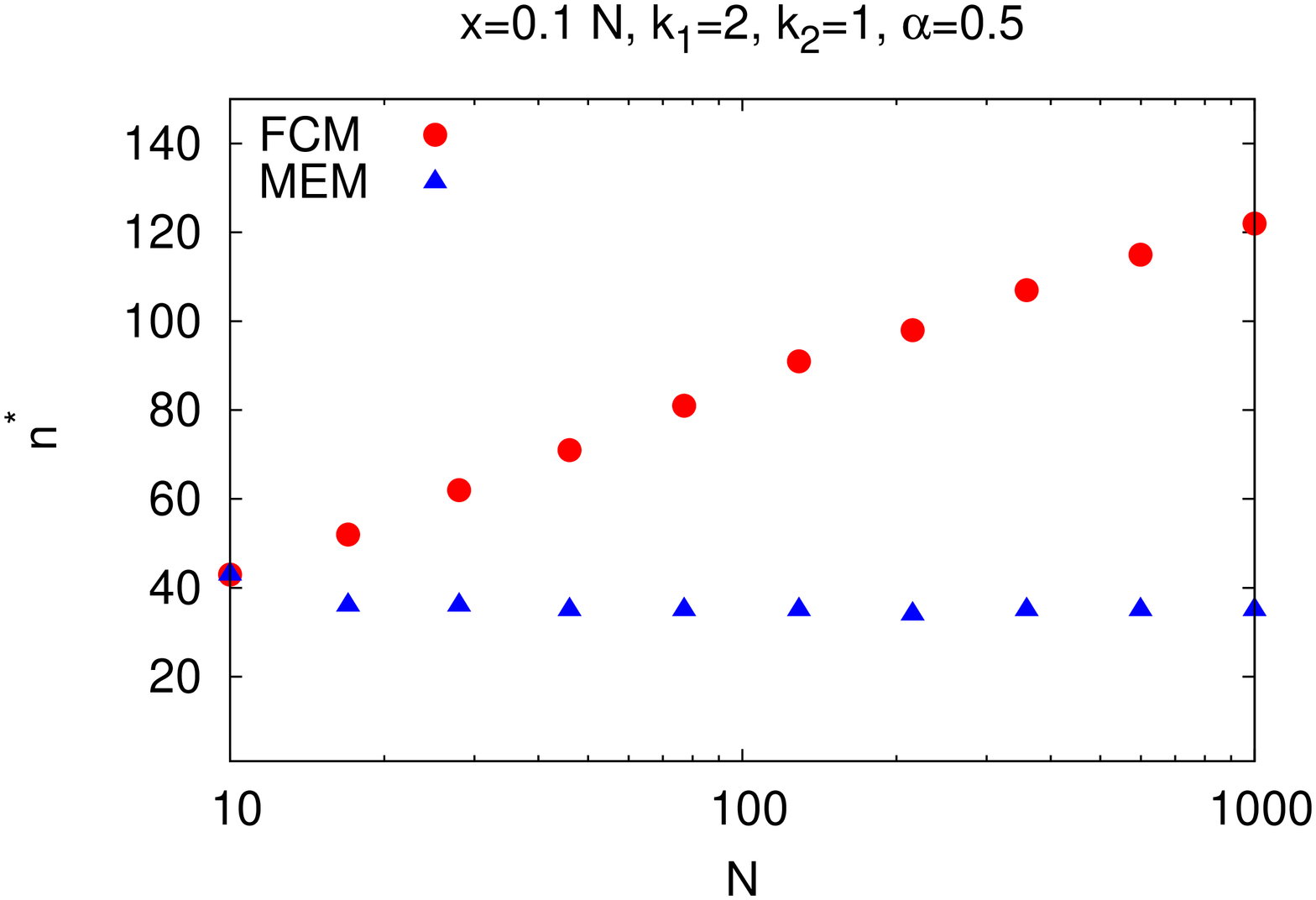}
 \caption{Convergence iteration $n^{*}$ as a function of size $N$ for artificial matrices generated 
 according to Model B described in \ref{complexity} and illustrated in Fig. \ref{fig:modelAresults}.
 The panels show that for the MEM, the dependence of convergence speed on system size strongly depends on the parameters chosen
 to construct the matrix.
 \label{fig:modelBresults}}
\end{figure}

In this section, we build artificial nested matrices to study the dependence of the convergence iteration of the metrics
on network size.
The convergence iteration $n^{*}$ is defined as the smallest iteration such that condition \eqref{convergence_criterion} holds.
We focus on perfectly nested matrices where the diagonal never crosses the empty region of the matrix,
as is the case for the real matrix showed in Fig. \ref{fig:matrix}.
We label the countries in order of increasing diversification and the products in order of decreasing ubiquity.
To generate the matrices, we use two models:
\begin{itemize}
 \item \emph{Model A:} This model has a single parameter
 $\alpha$ which determines the shape of the matrix.
 In order to have the same ratio $M/N$ as in the real data from world trade analyzed in the main text,
we set $M=5.48\,N$.
 For country $i$, we fill the elements corresponding to products $\alpha\in[1,1+ \floor*{M\,i^{\alpha}/N^{\alpha}}]$,
where $\alpha$ is a parameter of the matrix that determines the shape of the border which
separates the empty and the filled regions of the matrix,
and $ \floor*{x}$ denotes the largest integer smaller or equal than $x$.
We restrict our analysis to $\alpha<1$ which corresponds to matrices for which the diagonal does not cross the empty region. 
\item \emph{Model B:} This model has four parameters $x, \alpha, k_1, k_2$ which determine the shape of the matrix.
Fig. \ref{fig:modelB} shows an illustration of a matrix produced with model B.
Country $1$ has diversification equal to $x+k_1$. For the countries $i\in[2, \floor*{\alpha\,(N-1)}]$,
$d_{i+1}=d_i+k_1$ holds; for the remaining countries, $d_{i+1}=d_i+k_2$ holds.
For each value of $N$, the number of products $M$ is determined by the parameters of the model. 
\end{itemize}

Within this framework, we can study the dependence of the convergence speed of the metrics on network size for 
a given shape of the matrix's border.
For Model A, we find that the convergence speed of the MEM does not strongly depend on system size, as 
opposed to the convergence speed of the FCM which grows approximately as $\log{(N)}$ for sufficiently large $N$ (see Fig. \ref{fig:modelAresults}).
For Model B, we find again a logarithmic growth of the convergence iteration for the FCM for sufficiently large $N$,
whereas the behavior of the MEM can be very different with respect to that found for Model A (see Fig. \ref{fig:modelBresults}).
In particular, for some parameter settings the convergence of the MEM is slower than that of the FCM, as found in the real data.
Figs. \ref{fig:modelAresults} and \ref{fig:modelBresults} indicate that the convergence behavior of the MEM is strongly dependent on the details of
the border of the matrix $\mathcal{M}$, as opposed to the FCM which always exhibit asymptotic logarithmic dependence of $n^{*}$ on $N$.
We did not attempt to investigate the convergence behavior of the metric on alternative matrix models. 
We envision that suitable modifications of the equations that define the two metrics would mitigate
the dependence of convergence speed on system size; however, designing new metrics with improved convergence properties goes beyond the scope
of this article.

\section{Dividing the matrix $\mathcal{M}$ into submatrices based on fitness ratio}
\label{submatrices}

 \begin{figure}[h]
 \centering
 \includegraphics[height=0.45\columnwidth,  angle=270]{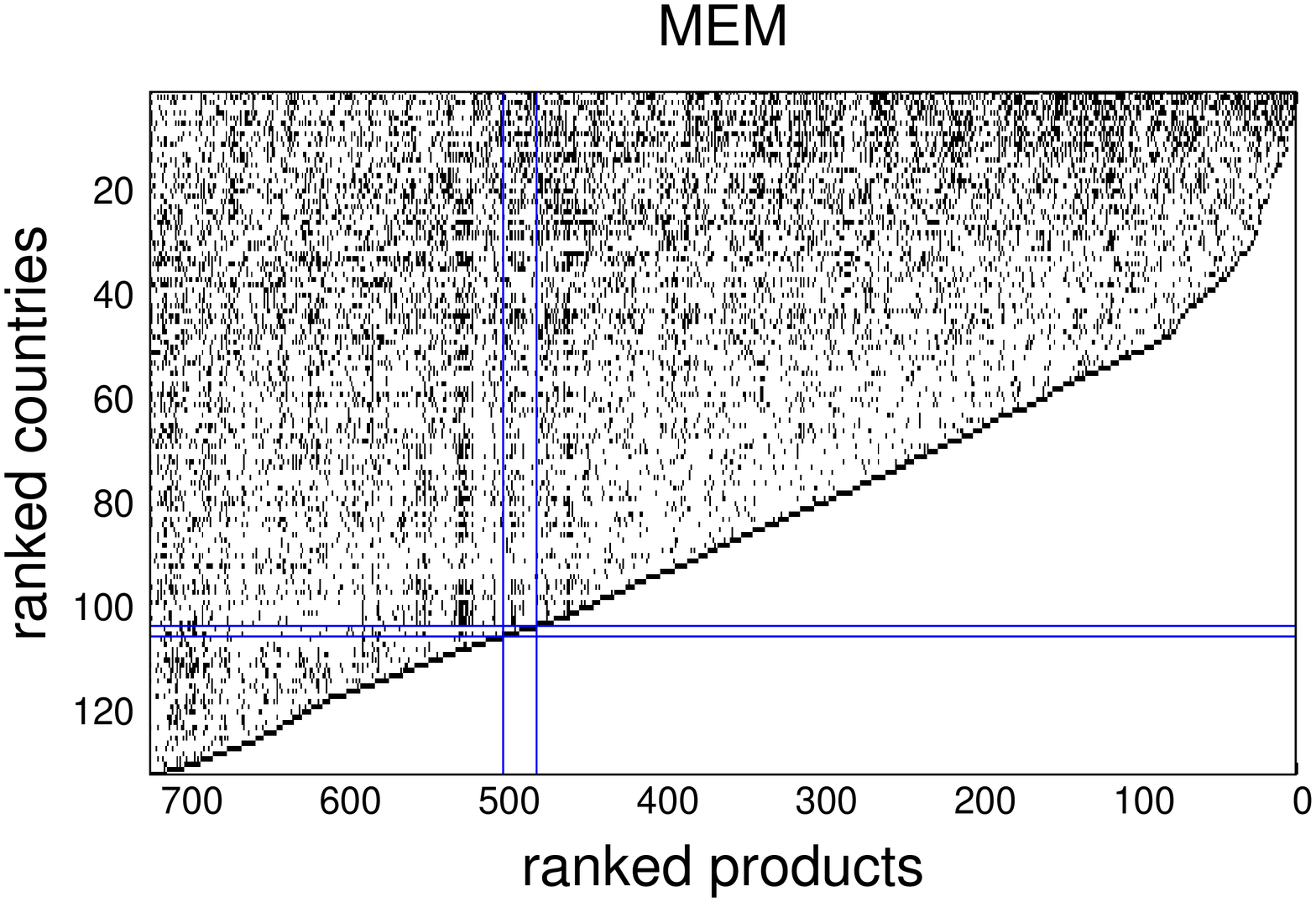}\includegraphics[height=0.45\columnwidth,  angle=270]{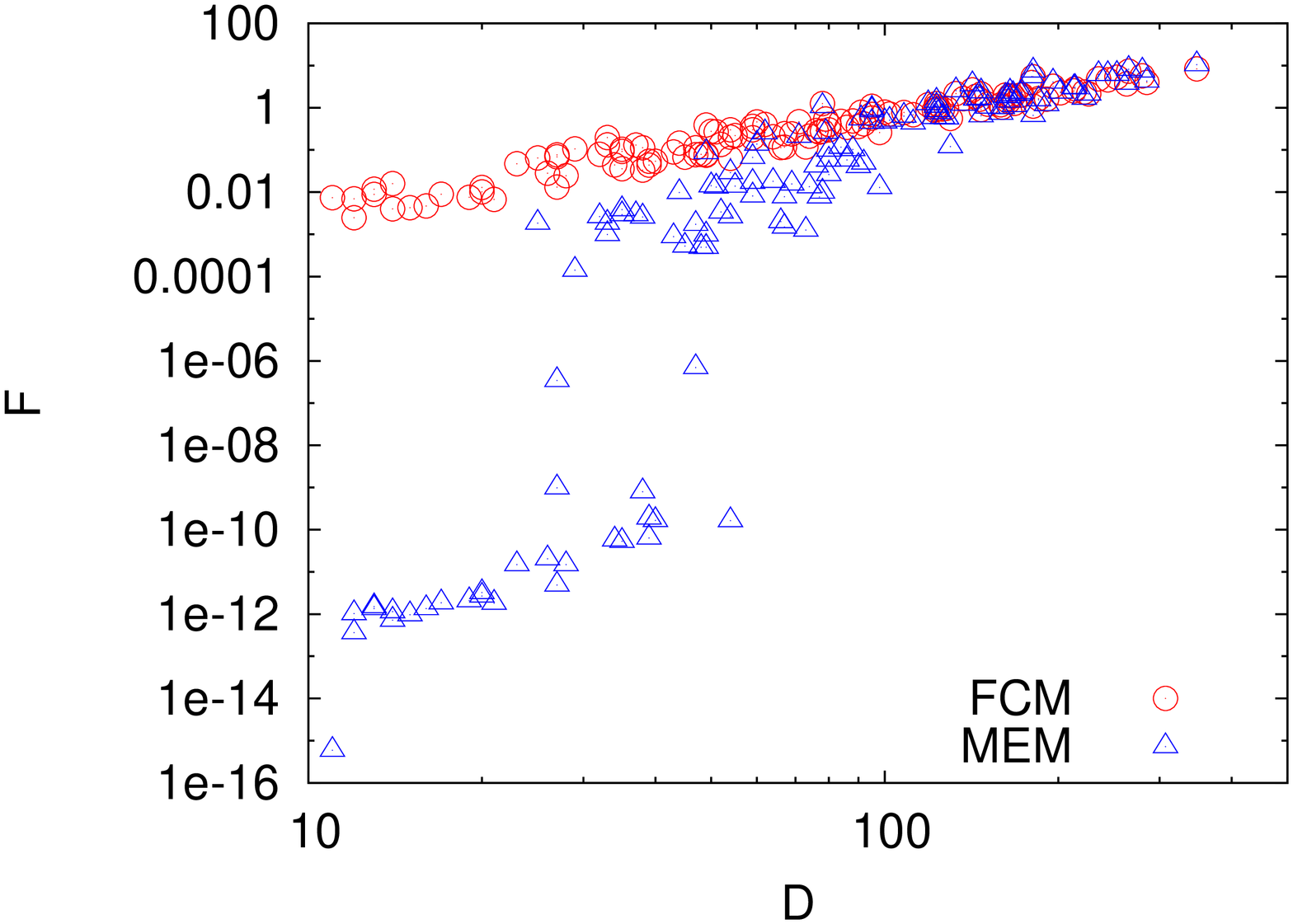}
 \caption{\emph{Left:} Country-product matrices resulting from the MEM (1996); horizontal and vertical blue lines separate
 groups of countries and products, respectively. \emph{Right:} Country fitness score $F$ vs. diversification $D$ for the FCM and the MEM.
 \label{fig:matrix_blocks}}
\end{figure}

When some fitness score ratios converge to zero, the matrix $\mathcal{M}$ can be separated into different groups of countries
such that the score ratios between countries within the same group are always larger than zero.
For the FCM, one or more fitness ratios converge to zero
when the diagonal of the matrix crosses the empty region of the matrix $\mathcal{M}$ (see section \ref{analytics_crossing}
and ref. \cite{pugliese2014convergence}).
For the MEM and for a perfectly nested matrix $\mathcal{M}$, one or more fitness ratios converge to zero
when the diversification gap between two countries $i+1$ and $i$ is equal or larger
than the maximum diversification gap of the lower ranked countries, as directly results from Eq. \eqref{solution_minimal}.
While the criterion for the MEM is not directly applicable to real matrices that are not perfectly nested,
we empirically observe in the dataset used for Fig. \ref{fig:matrix} that the fitness ratios of two pairs of countries converge to 
zero. As suggested in ref. \cite{pugliese2014convergence}, 
we can then separate the countries into three groups such that the fitness ratios are nonzero between any two countries that belong
to the same group.
The three resulting groups are composed of $103$, $2$ and $27$ countries, respectively (see Fig. \ref{fig:matrix_blocks}, left panel).
The right panel of Fig. \ref{fig:matrix_blocks} shows that the separation of countries into different groups
is signaled by discontinuous jumps in the relation between country MEM fitness and country diversification $D$,
which happens for $D<100$.
We emphasize that while the deviation between the trends observed for the FCM and the MEM is relatively small for highly diversified countries, 
it becomes wide for little diversified countries, which might be relevant for the study of the economic complexity dynamics of 
developing countries \cite{cristelli2015heterogeneous}.

\section*{Acknowledgements}

This work was supported by the EU FET-Open Grant No. 611272 (project Growthcom), by the Swiss National Science Foundation (Grant No. 200020-143272),
by the China Scholarship Council (CSC) scholarship and by the National Natural Science Fundation of China (Grant No. 61503140).
We wish to thank Alexandre Vidmer for his careful proofreading of the manuscript and his useful suggestions.
We acknowledge useful discussions with Mat{\'u}{\v{s}} Medo, Emanuele Pugliese, Andrea Zaccaria.

\bibliographystyle{elsarticle-num}

\end{document}